\documentclass[sigconf]{acmart}
\AtBeginDocument{%
  \providecommand\BibTeX{{%
    \normalfont B\kern-0.5em{\scshape i\kern-0.25em b}\kern-0.8em\TeX}}}

\copyrightyear{2020}
\acmYear{2020}
\setcopyright{iw3c2w3}
\acmConference[WWW '20]{Proceedings of The Web Conference 2020}{April 20--24, 2020}{Taipei, Taiwan}
\acmBooktitle{Proceedings of The Web Conference 2020 (WWW '20), April 20--24, 2020, Taipei, Taiwan}
\acmPrice{}
\acmDOI{10.1145/3366423.3380220}
\acmISBN{978-1-4503-7023-3/20/04}

\usepackage[framemethod=TikZ]{mdframed}
\usepackage{xspace} % not in TAPS
\usepackage{microtype} % not in TAPS
\usepackage{xcolor}
\usepackage{tkz-berge}
\usepackage{ordinalref} % for conference numbers in the bibliography

\usepackage{tikz}
\usetikzlibrary{shadows}
\usepackage{float}

\usepackage[noend]{algpseudocode}

\usepackage{multirow}
\usepackage{booktabs}

\usepackage{algorithm}
\newcommand{\code}[1]{\textsc{#1}}

\newcommand{\ouralg}{\textsc{QECC}}
\newcommand{\ourheur}{\textsc{QECC-heur}}

\DeclareMathOperator{\cost}{cost}
\newcommand{\balls}{\textsf{QwickCluster}}
\newcommand{\ballslow}{\textsf{SluggishCluster}}
\newcommand{\opti}{\func{OPT}}
\newcommand{\acc}{\textsf{ACC}}
\newcommand{\eps}{\varepsilon}
\newcommand{\func}[1]{\operatorname{#1}}
\DeclareMathOperator{\similarity}{sim}
\DeclareMathOperator{\cl}{\ell}

\newcommand{\mycomment}[1]{}
\DeclareMathOperator*{\argmax}{arg\,max}

\newenvironment{prooftext}[1]{\par\noindent{\bf Proof#1.}\quad}{\nopagebreak$\qed$\\}
\newenvironment{proofof}[1]{\begin{prooftext}{ of #1}}{\end{prooftext}}

\newtheorem{problem}{Problem}

\providecommand{\poly}{{\operatorname{poly}}}

\DeclareMathOperator*{\expect}{\mathbb{E}}

\DeclareRobustCommand{\calG}[0]{{\mathcal G}}

\DeclareRobustCommand{\calS}[0]{{\mathcal S}}

\newcommand{\squishlist}{\begin{list}{$\bullet$}
  { \setlength{\itemsep}{0pt}
     \setlength{\parsep}{3pt}
     \setlength{\topsep}{3pt}
     \setlength{\partopsep}{0pt}
     \setlength{\leftmargin}{1.5em}
     \setlength{\labelwidth}{1em}
     \setlength{\labelsep}{0.5em} } }
\newcommand{\squishend}{
  \end{list}  }

\newcommand{\spara}[1]{\smallskip\noindent{\bf #1}}

\begin{document}

\title{Query-Efficient Correlation Clustering}


\author{David Garc\'ia--Soriano}
\email{d.garcia.soriano@isi.it}
\affiliation{\institution{ISI Foundation} \city{Turin} \country{Italy}}

\author{Konstantin Kutzkov}
\email{kutzkov@gmail.com}
\affiliation{\institution{Amalfi Analytics} \city{Barcelona} \country{Spain}}

\author{Francesco Bonchi}
\email{francesco.bonchi@isi.it}
\affiliation{\institution{ISI Foundation, Turin, Italy}}
\affiliation{\institution{Eurecat, Barcelona, Spain}}

\author{Charalampos Tsourakakis}
\email{ctsourak@bu.edu}
\affiliation{\institution{Boston University} \country{USA}}

\renewcommand{\shortauthors}{D. Garc\'ia--Soriano, K. Kutzkov, F. Bonchi, and C. Tsourakakis}

\begin{abstract}
\emph{Correlation clustering} is arguably the most natural formulation of clustering. Given $n$ objects and a pairwise similarity measure, the goal is to cluster the objects so that, to the best possible extent, similar objects are put in the same cluster and dissimilar objects are put in different clusters.

A main drawback of correlation clustering is that
it requires as input the $\Theta(n^2)$ pairwise similarities. This is often infeasible to compute or even just to store.
In this paper we study \emph{query-efficient} algorithms for correlation clustering.  Specifically, we devise a  correlation clustering algorithm that, given a budget of $Q$
queries, attains a solution whose expected number of disagreements is at most $3\cdot \opti + O(\frac{n^3}{Q})$, where $\opti$ is the optimal cost for the instance.  Its running
time is $O(Q)$, and can be easily made non-adaptive (meaning it can specify all its queries at the outset and make them in parallel) with the same guarantees.
Up to constant factors, our algorithm yields a provably optimal trade-off between the number of queries $Q$ and the worst-case
error attained, even for adaptive algorithms.  

Finally, we perform an experimental study of our proposed method on both synthetic and real data,  showing the scalability and the accuracy of our algorithm.

\end{abstract}





\maketitle

\section{Introduction}
\label{sec:intro}
\emph{Correlation clustering}~\cite{correlation_clustering} (or \emph{cluster editing}) is a prominent clustering framework where   we are given a set $V=[n]$ and a symmetric pairwise similarity function~$\similarity: \binom{V}{2}
\rightarrow \{ 0,1\}$, where $\binom{V}{2}$ is the set of unordered pairs of elements of $V$. The goal is to cluster the items in such a way that, to
the best possible extent, similar objects are put in the same cluster and dissimilar objects are put in different clusters. Assuming that cluster identifiers are represented by natural numbers, a clustering $\cl$ is a function $\cl:V \rightarrow \mathbb{N}$, and each cluster is a maximal set of vertices sharing the same label. Correlation clustering aims at minimizing the following cost:

\begin{equation}
\label{equation:correlation-clustering} 	
\cost(\cl) =
 \sum_{\substack{(x,y) \in \binom{V}{2}, \\ \cl(x)=\cl(y)}} (1-\similarity(x,y)) \mspace{5.0mu} +\mspace{-15.0mu}
 \sum_{\substack{(x,y) \in \binom{V}{2}, \\ \cl(x)\not=\cl(y)}} \similarity(x,y).
%\mspace{-20.0mu} \sum_{\substack{(x,y) \in \binom{V}{2}, \\ \cl(x)=\cl(y)}} (1-\similarity(x,y)) \mspace{5.0mu} +
%\mspace{-15.0mu} \sum_{\substack{(x,y) \in \binom{V}{2}, \\ \cl(x)\not=\cl(y)}} \similarity(x,y).
\end{equation}

The intuition underlying the above problem definition is that
if two objects $x$ and $y$ are dissimilar and are assigned to the same cluster we should pay a cost of $1$, i.e., the amount of their dissimilarity. Similarly, if $x,y$ are similar and they are assigned to different clusters we should pay also cost $1$, i.e., the amount of their similarity $\similarity(x,y)$.
The correlation clustering framework naturally extends to non-binary, symmetric function, i.e.,  $\similarity: \binom{V}{2}\rightarrow [0,1]$. In this paper we focus on the binary case; the general non-binary case can be efficiently reduced to this case at a loss of only a constant factor in the approximation~{\cite[Thm. 23]{correlation_clustering}}.
The binary setting can be viewed very conveniently through graph-theoretic lenses: the $n$ items correspond to the vertices of a \emph{similarity graph} $G$, which is a complete undirected graph with edges labeled ``+'' or ``-''.
An edge $e$ causes a \emph{disagreement} (of \emph{cost} 1) between the similarity graph and a
clustering when it is a ``+'' edge connecting vertices in different clusters, or a ``--''
edge connecting vertices within the same cluster.
If we were given a \emph{cluster graph} \cite{cluster_editing}, i.e., a graph whose set of positive edges is the union
of vertex-disjoint cliques, we would be able to produce a perfect (i.e., cost 0) clustering simply by computing the connected components of the positive graph.
However, similarities will generally be inconsistent with one another, so incurring a certain cost is unavoidable.
Correlation clustering aims at minimizing such cost.
The problem may be viewed as  the task of finding the equivalence relation that most closely resembles a given symmetric relation.
The correlation clustering problem is
NP-hard~\cite{correlation_clustering, cluster_editing}.

Correlation clustering is particularly appealing for the task of clustering structured objects, where the similarity function is domain-specific. A typical application is clustering web-pages based on similarity scores: for each pair of pages we have a score between 0 and 1, and we would like to cluster the pages so that
pages with a high similarity score are in the same cluster, and pages with a low similarity score are in different clusters.
The technique is applicable to a multitude of problems in different domains, including duplicate detection and  similarity joins~\cite{duplicate_detection,corr_weighted}, spam detection~\cite{spam_filter, cc_tutorial}, co-reference resolution~\cite{correference}, biology~\cite{clustering_genes,BonchiGU13}, image segmentation~\cite{image_segmentation}, social networks~\cite{chromatic_clustering}, and clustering aggregation~\cite{clustering_aggregation}.
A key feature of correlation clustering is that it does not require the number of clusters as part of the input; instead it automatically finds the optimal number, performing model selection.

Despite its appeal, the main practical drawback of correlation clustering is the fact that, given $n$ items to be clustered,
$\Theta(n^2)$ similarity computations are needed to prepare the similarity graph that serves as input for the algorithm.  In addition to the obvious algorithmic cost involved with
$\Theta(n^2)$ queries, in certain applications  there is an additional type of cost that may render correlation clustering algorithms impractical.  Consider the following
motivating real-world scenarios.  In biological sciences, in order to produce a network of interactions between a set of biological entities (e.g., proteins), a highly trained
professional has to devote time and costly resources (e.g., equipment) to perform  tests between all ${n \choose 2}$ pairs of entities. In entity resolution, a task central to data
integration and data cleaning \cite{wang2012crowder}, a crowdsourcing-based approach performs queries to workers  of the form ``does the record $x$ represent the same entity as the
record $y$?''. Such queries to workers involve a monetary cost, so it is desirable to reduce their number. In both scenarios  developing clustering tools that use fewer than ${n \choose 2}$ queries is of major interest.  This is the main motivation behind our work. At a high level we answer the following question:

\smallskip

\begin{mdframed}[innerbottommargin=5pt,innertopmargin=1pt,innerleftmargin=6pt,innerrightmargin=6pt,backgroundcolor=gray!10,roundcorner=10pt]
\begin{problem}\label{prob1}
How to design a correlation clustering algorithm that outputs a good approximation in a \emph{query-efficient} manner: i.e., given a budget of $Q$ queries, the algorithm is allowed to learn the specific value of $\similarity(i, j) \in \{0,1\}$ for up to $Q$ queries $(i, j)$ of the algorithm's choice.
\end{problem}
\end{mdframed}

\spara{Contributions.} The main contributions of this work are summarized as follows:

\begin{itemize}
  \item We design a computationally efficient randomized algorithm $\ouralg$ that, given
a budget of $Q$ queries, attains a solution whose expected number of disagreements is at most $3\cdot \opti + O(\frac{n^3}{Q})$, where $\opti$ is the optimal cost of the
correlation clustering instance (Theorem~\ref{thm:ub}).
We can achieve this via a \emph{non-adaptive} algorithm (Theorem~\ref{thm:ub2}).%, meaning that all its queries can be specified at the outset and can thus be performed in parallel.

  \item We show (Theorem~\ref{thm:lb}) that up to constant factors, our algorithm is optimal, even for adaptive algorithms: any algorithm making $Q$ queries must make at least
$\Omega(\opti + \frac{n^3}{Q})$ errors.% for some instance.

  \item We give a simple, intuitive heuristic modification of our algorithm $\ourheur$ which helps reduce the error  of the algorithm in practice (specifically the recall of
        positive edges), thus partially
bridging the constant-factor gap between our lower and upper bounds.

  \item We present an experimental study of our two algorithms, compare their performance with a baseline based on affinity propagation, and study their sensitivity to parameters such
as graph size, number of clusters, imbalance, and noise.

\end{itemize}

\section{Related work}\label{sec:related}
We review briefly the related work that lies closest to our  paper.

\spara{Correlation clustering.}  The correlation clustering problem is NP-hard~\cite{correlation_clustering, cluster_editing} and, in its minimizing disagreements formulation used above, is also
APX-hard~\cite{cluster_qualitative}, so we do not expect a polynomial-time approximation scheme.
Nevertheless, there are constant-factor approximation algorithms~\cite{correlation_clustering,cluster_qualitative,balls}.
Ailon et al.~\cite{balls} present $\balls$, a
simple, elegant 3-approximation algorithm. They improve the approximation ratio to $2.5$ by utilizing an LP relaxation of the problem; the best approximation factor known to date
is $2.06$, due to Chawla et al.~\cite{near_opt}. The interested reader may refer to the extensive survey due to Bonchi, Garc\'ia-Soriano and Liberty \cite{cc_tutorial}.

\spara{Query-efficient algorithms for correlation clustering.} Query-efficient correlation clustering  has received less attention.  There exist two categories of algorithms, non-adaptive and adaptive. The
former choose their queries before-hand, while the latter can select the next query  based on the response to previous queries.

In an earlier preprint~\cite{local_corr} we initiated the study of query-efficient correlation clustering. Our work there focused on a stronger local model which requires answering cluster-id queries quickly, i.e., outputting a cluster label for each given vertex by querying at most $q$ edges per vertex.
Such a
$q$-query local algorithm %allows answering same-cluster queries with $2q$ edge queries, and
allows a global clustering of the graph with $q n$ queries; hence upper bounds for local
clustering imply upper bounds for global clustering, which is the model we consider in this paper. % (whereas lower bounds for global clustering imply lower bounds for local clustering).
The algorithm from~\cite{local_corr} is non-adaptive, and the upper bounds we present here (Theorems~\ref{thm:ub} and~\ref{thm:ub2}) may be
recovered by setting
$\epsilon = \frac nQ$ in~\cite[Thm. 3.3]{local_corr}.
A matching lower bound  was also proved in~\cite[Thm. 6.1]{local_corr}, but the proof therein applied only to non-adaptive algorithms.
In this paper we present a self-contained analysis of the algorithm from~\cite{local_corr} in the global setting (Theorems~\ref{thm:ub} and~\ref{thm:ub2}) and
strengthen the lower bound so that it applies also to adaptive
algorithms (Theorem~\ref{thm:lb}). Additionally, we perform an experimental study of the algorithm.

Some of the results from~\cite{local_corr} have been rediscovered  several years later (in a weaker form) 
    by Bressan, Cesa-Bianchi, Paudice, and Vitale~\cite{acc}. They study  the problem of query-efficient correlation clustering (Problem~1) in the adaptive setting, and provide a query-efficient algorithm, named $\acc$.
The performance guarantee they obtain in~\cite[Thm. 1]{acc} is asymptotically the same that had already been proven in~\cite{local_corr}, but it has worse constant factors and is attained via an
adaptive algorithm. They also modify the lower bound proof from~\cite{local_corr} to make it adaptive~\cite[Thm. 9]{acc}, and present some new results concerning the
cluster-recovery properties of the algorithm.%$\acc$ and exponential-time query-efficient algorithms.

%In terms of techniques, the $\acc$ algorithm from~\cite{acc} is very similar to ours, the only difference being the addition of
In terms of techniques, the only difference between our algorithm $\ouralg$ and the $\acc$ algorithm from~\cite{acc} is that the latter adds
a check that discards pivots when no neighbor
is found after inspecting a random sample of size $f(n-1)=Q/(n-1)$.      This additional check is
unnecessary from a theoretical viewpoint  (see Theorem~\ref{thm:ub}) and it has the disadvantage that it necessarily results in an adaptive algorithm.
%In addition, such check necessarily results in an adaptive algorithm.%, while our algorithm can be made non-adaptive (Theorem~\ref{thm:ub2}).
Moreover, the analysis of ~\cite{acc} is significantly more complex than ours, because they need to adapt the proof of the approximation guarantees of the $\balls$
algorithm from
~\cite{balls} to take into account the additional check, whereas we simply take the approximation guarantee as given and argue that stopping $\balls$ after $k$
pivots have been selected only incurs an expected additional cost of $n^2/(2k)$ (Lemma~\ref{lem:indep}).

\section{Algorithm and analysis}\label{sec:main}
%\enlargethispage{\baselineskip}
\mycomment{
Before presenting our algorithm,  we need to describe the \balls\ procedure of Ailon et al.~\cite{balls}. It selects a random pivot, creates a cluster with it and its positive
neighborhood, removes the cluster, and iterates on the induced subgraph remaining. Essentially, it finds a maximal independent set in the positive graph in random order.
In~\cite{balls}, the authors show that the expected cost of the clustering found by $\balls$ is at most $3\opti$.
\begin{algorithm}\label{alg:balls}
\begin{algorithmic}[0]
\Require $G = (V, E)$, a complete graph with ``+,-'' edge labels
\State $R \gets V$ \Comment Unclustered vertices so far
\While {$R \neq \emptyset$}
\State Pick a pivot $v$ from $R$ uniformly at random.
\State Output cluster $C = \{v\} \cup \Gamma_G^+(v) \cap R$.
\State $R \gets V \setminus C$
\EndWhile
\caption{\balls}
\end{algorithmic}
\end{algorithm}
}

Before presenting our algorithm, we describe in greater detail the elegant algorithm due to Ailon et al. \cite{balls} for correlation clustering, as it lies close to our proposed method.

       \spara{\balls\ algorithm.}  The \balls\ algorithm selects a random pivot $v$, creates a cluster with $v$ and its positive
neighborhood, removes the cluster, and iterates on the induced  remaining subgraph. Essentially it finds a maximal independent set in the positive graph in random order.
The elements in this set serve as cluster centers (pivots) in the order in which they were found.
In the pseudocode below, $\Gamma_G^+(v)$ %  = \{u \in V \mid (u, v) \in E^+ \}$ 
denotes the set of vertices to which there is a positive edge in $G$ from $v$.

\begin{algorithm}\label{alg:balls}
\begin{algorithmic}[0]
\Require $G = (V, E)$, a complete graph with ``+,-'' edge labels
\State $R \gets V$ \Comment Unclustered vertices so far
\While {$R \neq \emptyset$}
\State Pick a pivot $v$ from $R$ uniformly at random.
\State Output cluster $C = \{v\} \cup \Gamma_G^+(v) \cap R$.
\State $R \gets V \setminus C$
\EndWhile
\caption{\balls}
\end{algorithmic}
\end{algorithm}

\noindent When the graph is clusterable, $\balls$ makes no mistakes. In~\cite{balls}, the authors show that the expected cost of the clustering found by $\balls$ is at most
$3\opti$, where $\opti$ denotes the optimal cost.

%\spara{Preliminaries.}
%All our graphs are undirected and simple. For a vertex $v$, $\Gamma^+(v)$ is the set
%of positive edges incident with $v$; similarly define $\Gamma^-(v)$.
%We extend this notation to sets of vertices in the obvious manner.
%The  \emph{distance} between two graphs $G=(V,E)$ and $G'=(V,E')$ is  $|E \oplus E'|$.  Their fractional distance is their distance divided by $n^2$ (note this is in the interval $[0, 1/2)$).  Two graphs are \emph{$\eps$-close} to each other if their distance is at most most $\eps n^2$.  A \emph{$k$-clusterable} graph is a union of at most $k$ vertex-disjoint cliques.  A graph is \emph{clusterable} if it is $k$-clusterable for some $k$.

\spara{$\ouralg$.} Our algorithm $\ouralg$ (Query-Efficient Correlation Clustering) runs $\balls$ until the query budget $Q$ is complete, and then outputs singleton clusters for the remaining unclustered vertices.
\begin{algorithm}\label{alg:ours}
\begin{algorithmic}[0]
\Require $G = (V, E)$; query budget $Q$
\State $R \gets V$ \Comment Unclustered vertices so far
\While {$R \neq \emptyset \wedge Q \ge |R| - 1$}
\State Pick a pivot $v$ from $R$ uniformly at random.
\State Query all pairs $(v, w)$ for $w \in R\setminus\{v\}$ to determine $\Gamma_G^+(v) \cap R$.
\State $Q \gets Q - |R| + 1$
\State Output cluster $C = \{v\} \cup \Gamma_G^+(v) \cap R$.
\State $R \gets V \setminus C$
\EndWhile
\State Output a separate singleton cluster for each remaining $v \in R$.
\caption{\ouralg}
\end{algorithmic}
\end{algorithm}
The following subsection is devoted to the proof of our main result, stated next.
\begin{theorem}\label{thm:ub}
Let $G$ be a graph with $n$ vertices.
For any $Q > 0$, Algorithm $\ouralg$ finds a clustering of $G$ with expected cost at most $3\cdot \opti + \frac{n^3}{2 Q}$ making at most $Q$ edge queries.
It runs in time $O(Q)$ assuming unit-cost queries.
\end{theorem}

\subsection{Analysis  of $\ouralg$}
%\enlargethispage{\baselineskip}
For simplicity, in the rest of this section we will identify a complete ``+,-'' labeled graph $G$ with its graph of \emph{positive} edges $(V, E^+)$, so that
queries correspond to querying a pair of vertices for the existence of an edge. The set of (positive) neighbors of $v$ in a graph $G = (V, E)$ will be denoted $\Gamma(v)$; a similar
notation is used for the set $\Gamma(S)$ of positive neighbors of a set $S \subseteq V$.
The cost of the optimum clustering for $G$ is denoted $\opti$.
When $\ell$ is a clustering, $\cost(\ell)$ denotes the cost (number of disagreements) of this clustering,
     defined by~\eqref{equation:correlation-clustering} with $\similarity(x, y) = 1$ iff $\{x, y\} \in E$.

In order to analyze $\ouralg$, we need to understand how early stopping of $\balls$ affects  the accuracy of the clustering found.
For any non-empty graph~$G$ and pivot $v \in V(G)$, let $N_v(G)$ denote the subgraph of $G$
resulting from removing all edges incident to $\Gamma(v)$ (keeping all vertices).
Define a random sequence $G_0, G_1, \ldots$ of graphs by $G_0 = G$ and $G_{i+1} =
N_{v_{i+1}}(G_i)$, where $v_1, v_2, \ldots$ are chosen independently and uniformly at random from $V(G_0)$. Note
        that $G_{i+1} = G_i$ if at step $i$ a vertex is chosen for a second time.

 The following lemma is key:
\begin{lemma}\label{lem:del_edges}
Let $G_i$ have average degree $\tilde{d}$. When going from $G_i$ to $G_{i+1}$, the number of edges
decreases in expectation by at least $\binom{\tilde{d}+1}{2}$.
%, and the number of degree-0 vertices increases in expectation by at least $\tilde{d} + 1$.
\end{lemma}
\begin{proof}
Let $V = V(G_0)$, $E = E(G_i)$ and let $d_u = |\Gamma(u)|$ denote the degree of $u \in V$ in $G_i$.
%The claim on the number of degree-0 vertices is easy: the chosen pivot $v$ and all its neighbours will become vertices of degree 0 in $G_{i+1}$, so the average increase in the number of degree-0 vertices is at least $\expect_{v\in V(G_0)} (1 + d_v) = \tilde{d} + 1$.
%We now prove the claim on the number of edges.
Consider an edge $\{u, v\} \in E$. It is deleted if the chosen pivot $v_i$ is an element of $\Gamma(u)
\cup \Gamma(v)$ (which contains $u$ and $v$). Let $X_{uv}$ be the 0-1 random variable associated with this event, which occurs with probability
$$ \expect[X_{uv}] = \frac{|\Gamma(u) \cup \Gamma(v)|}{n} \ge \frac{1 + \max(d_u, d_v)}{n} \ge
\frac{1}{n} + \frac{d_u + d_v}{2n}. $$
Let $D = \sum_{u < v \mid {\{u,v\}} \in E} X_{uv}$ be the number of edges deleted (we assume an ordering of $V$ to avoid double-counting edges).
By linearity of expectation,
%\begin{small}
\begin{eqnarray*}
 \expect[D] &=& \sum_{\substack{u < v\\ \{u, v\} \in E}} \expect[X_{uv}] = \frac{1}{2} \sum_{\substack{u, v \in V \\ \{u, v\} \in E}} \expect[X_{uv}] \\
    &\ge& \frac{1}{2}  \sum_{\substack{u, v\\ \{u, v\} \in E}}\left(\frac{1}{n} + \frac{d_u+d_v}{2n}\right)\\
   &=&  \frac{\tilde{d}}{2} + \frac{1}{4n} \sum_{\substack{u, v\\ \{u, v\} \in E}} (d_u + d_v).
\end{eqnarray*}
%\end{small}
Now we compute
\begin{eqnarray*}
       \frac{1}{4n} \sum_{\substack{u, v\\ \{u, v\} \in E}} (d_u + d_v)
    &=& \frac{1}{2n} \sum_{\substack{u, v\\ \{u, v\} \in E}} d_u
    = \frac{1}{2n} \sum_{u} d_u^2 \\
    &=& \frac{1}{2} \expect_{u \sim V} [d_u^2]
    \ge \frac{1}{2} \left(\expect_{u \sim V} [d_u]\right)^2
  = \frac{1}{2} \tilde{d}^2,
\end{eqnarray*}
where in the last line, $\sim$ denotes uniform sampling and we used the Cauchy-Schwarz inequality.
%where we used the Cauchy-Schwarz inequality in the last line.
Hence
$ \expect[D] \ge \frac{\tilde{d}}{2} + \frac{\tilde{d}^2}{2} = \binom{\tilde{d}+1}{2}. $
\end{proof}

\begin{lemma}\label{lem:indep}
%Let $G = (V, E)$ be a graph and $Q$ be an ordered sample of $r$ independent vertices uniformly chosen from~$V$, with or without replacement.  Let $P = \Call{IndependentSet}{Q}$.  Then the expected number of edges of $G$ not incident with an element of $P \cup \Gamma(P)$ is less than $\frac{n^2}{2r}$.
Let $G$ be a graph with $n$ vertices and let $P = \{v_1, \ldots, v_r\}$ be the first $r$ pivots chosen by running $\balls$ on $G$.
Then the expected number of positive edges of $G$ not incident with an
element of $P \cup \Gamma(P)$
is less than $\frac{n^2}{2 (r + 1)}$.
\end{lemma}
\begin{proof}
Recall that at each iteration $\balls$ picks a random pivot from $R$. This selection is equivalent to picking a random pivot $v$ from the original set of vertices $V$ and
discarding it if $v \notin R$, repeating until some $v \in R$ is found, in which case a new pivot is added.
Consider the following modification of $\balls$, denoted $\ballslow$, which picks
a pivot $v$ at random from $V$ but always increases the counter $r$ of pivots found, even if $v \in R$ (ignoring the cluster creation step if $v \notin R$). % (but does not create a new cluster if $v \notin R$).
We can couple both algorithms into
a common probability space where each point $\omega$ contains a sequence of randomly selected vertices and each algorithm picks the next one in sequence.
\mycomment{
; $\ballslow$ always increases the
pivot counter $r$ while $\balls$ only increases it if $v_r \in R$, otherwise moves to the next pivot candidate in $\omega$.
}
For
any $\omega$, whenever the first $r$ pivots of $\ballslow$ are $S = (v_1, \ldots, v_r)$, then the first
$r'$ pivots of $\balls$ are the sequence $S'$ obtained from $S$ by removing previously appearing elements, where $r' = |S'|$. Hence $|V \setminus (S \cup \Gamma(S))| = |V
\setminus (S' \cup \Gamma(S'))|$ and $r' \le r$.
Thus the number of edges not incident with the first $r$ pivots and their neighbors in $\ballslow$ stochastically dominates
the number of edges not incident with the first $r$ pivots and their neighbors in $\ballslow$, since both numbers are decreasing with $r$.

Therefore it is enough to prove the claim for $\ballslow$.
Let $n = |V(G_0)|$ and define $\alpha_i \in [0,1]$  by
$ \alpha_i = \frac{2 \cdot |E(G_i)|}{n^2} .$
%Let $\widetilde{V}(G) = \{v \in V(G) \mid \deg(v) > 0 \}$ and define the ``actual size'' $s(G)$ of a
%graph by $ s(G) = 2 \cdot |E(G)| + |\widetilde{V}(G)| = \sum_{v \in \widetilde{V}(G)} (1 + \deg(v)). $
%Let $n = |V(G_0)|$ and define $\alpha_i \in [0,1]$ by {
%$\alpha_i = \frac{s(G_i)}{n^2}.$ % and $$\beta_i = \frac{|\widetilde{V}(G_i)|}{n} \ge \alpha_i.$$
%}
We claim that
for all $i \ge 1$ the following inequalities hold:
\begin{align}
%\expect[ \beta_i \mid v_1,\ldots,v_{i-1}] &\le \beta_{i-1} - \alpha_{i-1},\label{eq:beta_i}\\
\expect[ \alpha_i \mid G_0,\ldots,G_{i-1}] &\le \alpha_{i-1} (1 -
        \alpha_{i-1})\label{eq:cond_alpha_i},\\
\expect[ \alpha_i ] &\le \expect[\alpha_{i-1}] (1 - \expect[\alpha_{i-1}]) \label{eq:alpha_i},\\
\expect[ \alpha_i ] &< \frac{1}{i+1} \label{eq:alpha_i2}.
\end{align}
Indeed, $G_{i}$ is a random function of $G_{i-1}$ only, and
the average degree of $G_{i-1}$ is $\widetilde{d}_{i-1} = \alpha_{i-1} n$ so, by Lemma~\ref{lem:del_edges},
    $$\expect[ 2 \cdot |E(G_{i})| \mid G_{i-1} ] \le \alpha_{i-1} n^2 - 2 \cdot \frac{1}{2} \widetilde{d}_{i-1}^2 = n^2 \alpha_{i-1} (1 - \alpha_{i-1}),$$
    proving~\eqref{eq:cond_alpha_i}.
    Now~\eqref{eq:alpha_i} now follows from Jensen's inequality: since
     $$ \expect[\alpha_i] = \expect\big[ \expect[\alpha_i \mid G_0, \ldots, G_{i-1}] \big] \le \expect[
     \alpha_{i-1} (1 - \alpha_{i-1})] $$
and the function $g(x) = x (1 - x)$ is concave in $[0,1]$, we have
$$ \expect[\alpha_i] \le \expect[g(\alpha_{i-1})] \le g(\expect[\alpha_{i-1}]) = \expect[\alpha_{i-1}] (1 -
        \expect[\alpha_{i-1}]). $$

Finally we prove $\expect[\alpha_i] < 1/(i+1) \; \forall i\ge1$. %, which implies the conclusion of the lemma. We know that
For $i = 1$, we have:
$$\expect[\alpha_1] \le g(\alpha_0) \le  \max_{x \in [0,1]} g(x) = g\left(\frac{1}{2}\right) = \frac{1}{4} < \frac{1}{2}.$$
For $i > 1$, observe that $g$ is increasing
on $[0, 1/2]$ and
$$ g\left(\frac{1}{i}\right) = \frac{1}{i} - \frac{1}{i^2} \le \frac{1}{i} - \frac{1}{i (i + 1)} = \frac{1}{i + 1}, $$
so~\eqref{eq:alpha_i2} follows from~\eqref{eq:alpha_i} by induction on $i$:
$$\expect[ \alpha_{i-1} ] < \frac{1}{i} \implies  \expect[ \alpha_i ] \le g\left(\frac{1}{i}\right) \le \frac{1}{i + 1}. $$
Therefore $\expect[ |E(G_{r})| ] = \frac{1}{2} \expect[ \alpha_r ] n^2 \le \frac{n^2}{2 (r + 1)}$, as we wished to show.
\end{proof}

We are now ready to prove Theorem~\ref{thm:ub}:

\smallskip 

\begin{proofof}{Theorem~\ref{thm:ub}}
Let $\opti$ denote the cost of the optimal clustering of $G$ and let $C_r$ be a random variable denoting the clustering obtained by stopping $\balls$ after $r$ pivots are
found (or running it to completion if it finds $r$ pivots or less), and putting all unclustered
vertices into singleton clusters.
Note that whenever $C_i$ makes a mistake on a negative edge, so does $C_j$ for $j \ge i$;
on the other hand, every mistake on a positive edge by $C_i$ is either a mistake by $C_j$ ($j \ge i$) or the edge is not incident to any of the vertices clustered in the first $i$
rounds. By Lemma~\ref{lem:indep}, there are at most $\frac{n^2}{2 (i + 1)}$ of the latter in expectation.
Hence $\expect[ \cost(C_i) ] - \expect[  \cost(C_n) ] \le \frac{n^2}{2(i + 1)}$.

Algorithm $\ouralg$ runs for $k$ rounds, where $k \ge \lfloor \frac{Q}{n - 1} \rfloor > \frac{Q}{n} - 1$ because each pivot uses $|R| - 1 \le n - 1$ queries.
Then $$\expect[ \cost(C_k) ] - \expect[ \cost(C_n) ] < \frac{n^2}{2(k + 1)} < \frac{n^3}{2 Q}.$$
On the other hand, we have $\expect[ \cost(C_n) ] \le 3 \cdot \opti$ because of the expected 3-approximation guarantee of $\balls$ from~\cite{balls}.
Thus $\expect[ \cost(C_k) ] \le 3 \opti + \frac{n^3}{2 Q}$, proving our approximation guarantee.

Finally, the time spent inside each iteration of the main loop is dominated by the time spent making queries to vertices in $R$, since this number also bounds
the size of the cluster found. Therefore the running time of $\ouralg$ is $O(Q)$.
\end{proofof}

%\spara{On adaptivity.}
\subsection{A non-adaptive algorithm.}
%It is interesting to note that adaptivity does not help for this problem.
Our algorithm $\ouralg$ is adaptive in the way we have chosen to present it: the queries made when picking a second pivot depend on the result of the queries made for
the first pivot.
However, this is not necessary: we can instead query for the neighborhood of a random sample $S$
of size $\frac{Q}{n-1}$. If we use the elements of $S$ to find pivots,
the same  analysis shows that the output of this variant meets the exact same error bound of $3\opti + n^3/(2Q)$.
For completeness, we include pseudocode for the adaptive variant of $\ouralg$ below (Algorithm 3).

In practice the adaptive variant we have presented in Algorithm~2 will run closer to the query budget, choosing more pivots and reducing the error somewhat
below the theoretical bound,
because it does not ``waste'' queries between a newly found pivot and the neighbors of previous pivots.
Nevertheless, in settings where the similarity computations can be performed in parallel, it may become advantageous to use Non-adaptive $\ouralg$.
Another benefit of the non-adaptive variant is that it gives a one-pass streaming algorithm for correlation clustering that uses only $O(Q)$ space and processes edges in
arbitrary order.

\begin{theorem}\label{thm:ub2}
For any $Q > 0$, Algorithm Non-adaptive $\ouralg$ finds a clustering of $G$ with expected cost at most $3\cdot \opti + \frac{n^3}{2 Q}$ making at most $Q$ non-adaptive edge queries.
It runs in time $O(Q)$ assuming unit-cost queries.
\end{theorem}
\begin{proof}
The number of queries it makes is $ S = (n-1)+(n-2)+\ldots(n-k) = \frac{2n-1-k}{2} k \le Q. $
Note that $\frac{n-1}{2} k \le S \le Q \le (n - 1)k$.
The proof of  the error bound proceeds exactly as in the proof of Theorem~\ref{thm:ub} (because $k \ge \frac{Q}{n-1}$).
The running time of the querying phase of Non-adaptive $\ouralg$ is $O(Q)$ and,
assuming a hash table is used to store query answers, the expected running time of the second phase is bounded by $O(n k) = O(Q)$,
         because $k \le \frac{2Q}{n-1}$.
\end{proof}
Another interesting consequence of this result (coupled with our lower bound, Theorem~\ref{thm:lb}), is that adaptivity does not help for correlation clustering (beyond possibly a constant factor), in stark contrast to
other problems where an exponential separation is known between the query complexity of adaptive and non-adaptive algorithms (e.g.,~\cite{fi,obdds_lb}).

\algnewcommand{\LineComment}[1]{\State \(\triangleright\) #1}

\begin{algorithm}\label{alg:nonadaptive}
\begin{algorithmic}[0]
\Require $G = (V, E)$; query budget $Q$
\State $k \gets \max\{ t \le n \mid (2n-1-t) t \le 2Q \}.$
\State Let $S = (s_1, \ldots, s_k)$ be a uniform random sample from $V$
\State \;(with or without replacement)
\State
\LineComment Querying phase: find $\Gamma_G^+(v)$ for each $v \in S$
\For { each $v \in S$, $w \in V$, $v < w$ }
\State Query $(v, w)$
\EndFor
\State
\LineComment Clustering phase
\State $R \gets V$
\State $i \gets 1$
\While {$R \neq \emptyset \wedge i \le k$}
    \If {$s_i \in R$}
        \State Output cluster $C = \{s_i\} \cup \Gamma_G^+(s_i) \cap R$.
        \State $R \gets V \setminus C$
    \EndIf
    \State $i \gets i + 1$
\EndWhile
\State Output a separate singleton cluster for each remaining $v \in R$.
\caption{Non-adaptive \ouralg}
\end{algorithmic}
\end{algorithm}

\section{Lower bound}\label{sec:lb}
In this section we show that $\ouralg$ is essentially optimal: for any given budget of queries, no algorithm (adaptive or not) can find a solution better than that of $\ouralg$ by more than a constant factor.

\begin{theorem}\label{thm:lb}
%For any $c \ge 1$ and $n < T \le \frac{n^2}{100 c}$, any algorithm finding a clustering with expected cost at most $c \cdot \opti + T$ must make at least $\Omega( \frac{n^3}{ T c^2})$ adaptive edge similarity queries.
For any $c \ge 1$ and $T$ such that $8n < T \le \frac{n^2}{2048 c^2}$, any algorithm finding a clustering with expected cost at most $c \cdot \opti + T$ must make at least $\Omega( \frac{n^3}{ T c^2})$ adaptive edge similarity queries.
\end{theorem}
Note that this also implies that any purely multiplicative approximation guarantee needs $\Omega(n^2)$ queries (e.g. by taking $T = 10n$).
\begin{proof}
Let $\epsilon = \frac{T}{n^2}$; then $\frac{1}{n} < \epsilon \le \frac{1}{2048 c^2}$.
By Yao's minimax principle~\cite{yao_minimax}, it suffices to produce a distribution $\calG$ over graphs with the following properties:
\begin{itemize}
    \item the expected cost of the optimal clustering of $G \sim \calG$ is $\expect[\opti(G)] \le \frac{\eps n^2}{c}; $
    \item for any {deterministic} algorithm making fewer than $L/2 = \frac{n}{2048 \eps c^2}$ queries, the expected cost (over $\calG$) of the clustering produced exceeds $2\eps
    n^2 \ge c \cdot \expect[\opti(G)] + T$.
\end{itemize}

Let $\alpha = \frac{1}{4 c}$ and $k = \frac{1}{32 c \eps}$. % and $l = \frac{k^2 \eps n}{3} \ge L$.
We can assume that $c$, $k$ and $\alpha n / k$ are integral (here we use the fact that $\eps > 1/n$).
Let $A = \{1, \ldots, (1 - \alpha) n\}$ and
$B  = \{(1-\alpha) n + 1, \ldots, n\}$.

Consider the following distribution $\calG$ of graphs:
partition the vertices of $A$ into
exactly $k$ equal-sized clusters $C_1, \ldots, C_k$. The set of positive edges will be the union of the cliques
defined by $C_1, \ldots, C_k$, plus edges joining each vertex $v \in B$ to
all the elements of $C_{r_v}$ for a randomly chosen $r_v \in [k]$; $r_v$ is chosen independently of $r_w$ for all $w \neq v$.

Define the \emph{natural clustering} of a graph $G \in \calG$ by the classes $C_i' = C_i \cup \{ v \in B\mid r_v = i\}$ ($i \in [k]$). 
We view $N$ also as a graph formed by a disjoint union of the $k$ cliques 
determined by $\{C'_i\}_{i\in[k]}$.
This clustering will have a few disagreements because of the negative edges between different
vertices $v, w\in B$ with $r_v = r_{w}$. For any pair of distinct elements $v, w \in B$, this happens with probability $1/k$. The cost of the optimal clustering of $G$ is bounded by that of
the natural clustering $N$, hence
$$ \expect [\opti]  \le \expect[\cost(N)] =  \frac{\binom{\alpha n}{2}}{k} \le \frac{\alpha^2 n^2}{2k} = \frac{\eps}{c} n^2. $$

We have to show that any algorithm making $< L/2$ queries to graphs drawn from~$\calG$
produces a clustering with expected cost larger than $2 \eps n^2$.
Since all graphs in $\calG$ induce the same subgraphs on $A$ and $B$ separately, we can assume
without loss of generality that the algorithm queries only edges between $A$ and $B$. 
Note that the neighborhoods  in $G$ of every pair of vertices from the same $C_i$ are the same:
$\Gamma^+_G(u) = \Gamma^+_G(v)$  and
$\Gamma^-_G(u) = \Gamma^-_G(v)$ 
for all $u, v \in C_i$, $i \in [k]$; moreover, $u$ and $v$ are joined by a positive edge.
Therefore, if $u, v \in C_i$ but the algorithm assigns $u$ and $v$ to different clusters, either moving $u$  to $v$'s cluster or $v$ to $u$'s cluster will not decrease the cost.
%Furthermore, it is never optimal to put a singleton cluster for any $u \in B$; assigning $u$ to one of the clusters $C_1, \ldots, C_k$ instead will never increase the cost.
All in all, we can assume that the algorithm outputs $k$ clusters $C'_1, \ldots, C'_k$ with $C_i \subseteq C'_i$ for all $i$, plus (possibly) some clusters 
$C'_{k+1}, \ldots, C'_{k'}$ ($k'\ge k$) involving only elements
of $B$.

For $v \in B$, let $s_v \in [k']$ denote the cluster that the algorithm assigns $v$ to.
%; $s_v=0$ means that $v$ is not assigned to any of $C'_1, \ldots, C'_k$. 
% and let $G_v$ denote the event that the algorithm queries some pair $(u, v)$ with $u \in C_{r_v}$.
For every $v \in B$,
let $G_v$ denote the event that the algorithm queries $(u, v)$ for some $u \in C_{r_v}$ and, whenever $G_v$ does not hold,
    let us add a ``fictitious'' query to the algorithm between $v$ and some arbitrary element of $C_{s_v}$. This ensures that whenever $r_v = s_v$, the last query
of the algorithm verifies its guess and returns 1 if the correct cluster has been found. This adds at most $|B| \le n \le \frac{L}{2}$ queries in total.
Let $Q_1, Q_2, \ldots, Q_z$ be the (random) sequence of queries issued by the algorithm and let $i^v_1, i^v_2, \ldots, i^v_{T_v}$ be the indices of those queries involving a fixed vertex $v \in
B$.
Note that $r_v$ is independent of the response to all queries not involving $v$ and, conditioned on the result of all queries up to time $t < i_{T_v}$, $r_v$ is uniformly distributed
among the set $\{ i \in [k] \mid (Q_j \notin C_i \forall j < t) \}$, whose size is upper-bounded by $k-t+1$.
Therefore
$$ \Pr[ Q_{i^v_t} \in C_{r_v} \mid Q_1, \ldots, Q_{i^v_{t-1}} ] \le \frac{1}{k - t + 1}, $$
which becomes an equality  if the algorithm does not query the same cluster twice.
It follows by induction that
\begin{equation}\label{eq:blah}
\Pr[ \{Q_{i^v_1}, \ldots, Q_{i^v_t}\} \cap  C_{r_v} \neq \emptyset ] \le \frac{t}{k}.
\end{equation}

Let $M_v$ be the event that the algorithm makes more than $k/2$ queries involving $v$. %, and $G_v$ the event that the algorithm queries $(u, v)$ for some $u \in C_{r_v}$.
        The event $r_v = s_v$ is equivalent to $G_v$, i.e., the event 
$\{Q_{i^v_1}, \ldots, Q_{i^v_{T^v}}\} \cap  C_{r_v} \neq \emptyset$,
because of our addition of one fictitious query for $v$.
We have
%$$ \Pr[ r_v = s_v ] \le \Pr[ M_v ] + \Pr[ G_v \wedge \overline{M}_v ] + \Pr[ r_v = s_v \wedge \overline{M}_v \wedge \overline{G}_v ].$$
$$ \Pr[ r_v = s_v ] = 
\Pr[G_v] \le \Pr[ M_v ] + \Pr[ G_v \wedge \overline{M}_v ].$$
In other words, either the algorithm makes many queries for $v$, or it hits the correct cluster with few queries.
(Without fictitious queries, we would have to add a third term for the probability that the algorithm picks by chance the correct $s_v$.)
%(recall that we operate under the assumption that the algorithm's last query verifies its guess if no edge has been found.)
We will use the first term $\Pr[ M_v ]$ to control the expected query complexity. The second term, $\Pr[G_v \wedge \overline{M}_v]$, is bounded by $\frac{1}{2}$ by~\eqref{eq:blah}
because $T_v \le k/2$ 
whenever $\overline{M}_v$ holds.
Hence
$$ \Pr[ r_v \neq s_v ] \ge \frac{1}{2} - \Pr[ M_v ],$$

so 
$$
\expect[ |\{v \in B \mid r_v \neq s_v \}| ] =  \sum_{v \in B} \Pr[ r_v \neq s_v ]
                                                     \ge \frac{\alpha n}{2} - \left(\sum_{v \in B} \Pr[ M_v ]\right).$$
Each vertex $v \in B$ with
$s_v \neq r_v$,
causes disagreements with all of $C_{r_v}\subseteq C'_{r_v}$ and $C_{s_v}\subseteq C'_{r_v}$, introducing at least $2 |A| / k \ge n / k$ new disagreements.

If we denote by $X$ the cost of the clustering found and by $Z$ the number of queries made, we have
\begin{align*}
\expect[X] &\ge \frac{n}{k} \expect[ |\{v \in B \mid r_v \neq s_v \}| ] \\
        &\ge \frac{\alpha n^2}{2k} - \frac{n}{k} \left(\sum_{v \in B} \Pr[ M_v ]\right)\\
        &= {4\epsilon n^2} - \frac{n}{k} \left(\sum_{v \in B} \Pr[ M_v ]\right).
\end{align*}
In particular, if $\expect[X] \le 2\epsilon n^2$, then we must have
       $$ \sum_{v \in B} \Pr[ M_v ] \ge \frac{2\epsilon n^2}{n/k} = {2  \epsilon n k} = \frac{n}{16 c} .$$
But then we can lower bound the expected number of queries by
$$ \expect[Z] \ge \frac{k}{2}  \sum_{v \in B} \Pr[ M_v ] \ge \frac{n k}{32 c} = \frac{n}{1024 c^2 \epsilon} = L = \frac{n^3}{1024 c^2 T},$$
of which at most $L/2$ are the fictitious queries we added. This completes the proof.

\end{proof}

\section{A practical improvement}\label{sec:improved_algo}
As we will see in Section~\ref{sec:experiments}, algorithm $\ouralg$, while provably optimal up to constant factors, sometimes returns solutions with poor recall of positive edges
when the query budget is low. Intuitively, the reason is that, while picking a random pivot works in expectation, sometimes a low-degree pivot is chosen and all $|R|-1$ queries are spent querying
its neighbors, which may not be worth the effort for a small cluster when the query budget is tight. To entice the algorithm to choose higher-degree vertices (which would also improve the recall), we propose to bias it so that pivots are chosen with probability
proportional to their positive degree in the subgraph induced by $R$. The conclusion of Lemma~\ref{lem:indep} remains unaltered in this case, but whether this change preserves the
approximation guarantees from ~\cite{balls} on which we rely is unclear.
In practice, this heuristic modification consistently
improves the recall on all the tests we performed, as well as the total number of disagreements in most cases.

\mycomment{
We cannot afford to compute the degree of each vertex with a small number of queries, but
the following scheme is easily seen to choose each vertex $u \in R$ with probability $(1 + d_u) / (|R| + 2E)$, where $d_u$ is the degree of $u$ in the subgraph $G[R]$ induced by $R$ and $E$ is the
total number of edges in $G[R]$:
\begin{enumerate}
    \item Pick random pairs of vertices $(u, v) \in R \times R$ until $u = v$ or an edge $(u, v) \in E$ is found;
    \item Select the first  endpoint $u$ of this edge as a pivot.
\end{enumerate}
The reason we sample with probability proportional to $1+d_u$ instead of $d_u$ is to avoid infinite loops when $G[R]$ contains no edges.
}
We cannot afford to compute the degree of each vertex with a small number of queries, but
the following scheme is easily seen to choose each vertex $u \in R$ with probability $d_u / (2E)$, where $d_u$ is the degree of $u$ in the subgraph $G[R]$ induced by $R$, and $E>0$ is the
total number of edges in $G[R]$:
\begin{enumerate}
    \item Pick random pairs of vertices to query $(u, v) \in R \times R$ until an edge $(u, v) \in E$ is found;
    \item Select the first endpoint $u$ of this edge as a pivot.
\end{enumerate}
When $E = 0$, this procedure will simply run out of queries to make.

%Indeed, $u$ is chosen uniformly at random from $R$, and the probability that another randomly chosen vertex is a neighbor is exactly $d_u/(|R| - 1)$, where $d_u$ denotes the degree of $u$. It follows that this scheme

Pseudocode for $\ourheur$ is shown below.
\begin{algorithm}\label{alg:our_heur}
\begin{algorithmic}[0]
\Require $G = (V, E)$; query budget $Q$
\State $R \gets V$ \Comment Unclustered vertices so far
\While {$|R| > 1 \wedge Q \ge |R| - 1$}
\State Pick a pair $(u, v)$ from $R \times R$ uniformly at random.
\State \If { $u \neq v$ }
\State Query $(u, v)$
\State $Q \gets Q - 1$
%\State \If {$u = v$ \textbf{or} $(u, v) \in E$ }
\State \If { $(u, v) \in E$ }
\State Query all pairs $(v, w)$ for $w \in R\setminus \{u, v\}$ 
\State \; to determine $\Gamma_G^+(v) \cap R$.
\State $Q \gets Q - |R| + 2$.
\State Output cluster $C = \{v\} \cup \Gamma_G^+(v) \cap R$.
\State $R \gets V \setminus C$
\EndIf
\EndIf
\EndWhile
\State Output a separate singleton cluster for each remaining $v \in R$.
\caption{\ourheur}
\end{algorithmic}
\end{algorithm}

\section{Experiments}\label{sec:experiments}
In this section we present the results of our experimental evaluations of $\ouralg$ and $\ourheur$, on both synthetic and real-world graphs.
We view an input graph as defining the set of positive edges; missing edges are interpreted as negative edges.

{%\scriptsize
\begin{table*}[t]
    \centering
    \caption{Dataset characteristics: name, type, size and ground truth error measures.}
    \label{tab:datasets}
    \begin{tabular}{lrrrrrrrrrr}
    \toprule
    Dataset  & Type & |V| & |E| & \# clusters  & GT cost & GT precision & GT recall \\
%        \\   &     &     &         densest subgraph    &     &\\
    \midrule
    \textsc{S(2000,20,0.15,2)} & synthetic & 2,000 & 104,985 & 20 & 30,483 & 0.859 & 0.852 \\
    \textsc{Cora}    &real & 1,879 & 64,955 & 191 & 23,516 & 0.829 & 0.803 \\
    \textsc{Citeseer} &real & 3,327 & 4,552 & - & - & - & -\\
    \textsc{Mushrooms} &real & 8,123 & 18,143,868 & 2 & 11,791,251 & 0.534 & 0.683 \\
    \bottomrule
    \end{tabular}
\end{table*}
}

\subsection{Experimental setup}
\spara{Clustering quality measures.}
We evaluate the clustering $S$ produced by $\ouralg$ and $\ourheur$ in terms of total \emph{cost} (number of disagreements),
\emph{precision} of positive edges
(ratio between the number of positive edges between pairs of nodes clustered together in $S$ and the total number of pairs of vertices clustered together in $S$),
 and
\emph{recall}    of positive edges
(ratio between the number of positive edges between pairs of nodes clustered together in $S$ and the total
 number of positive
            edges in $G$).
    Although our algorithms have been designed to minimize total cost, we deem it important to consider precision and recall values to detect extreme situations in which, for example, a graph is
    clustered into $n$ singletons clusters which,  if the graph is very sparse, may have small cost, but very low recall.
All but one of the graphs $G$ we use are accompanied with a ground-truth clustering (by design in the case of synthetic graphs), which we compare against.

\spara{Baseline.}
As $\ouralg$ is the first query-efficient algorithm for correlation clustering, any baseline must be based on another clustering method.
We turn to affinity propagation methods, in which  a matrix of similarities (affinities) are given as input, and then messages about the ``availability'' and
``responsibility'' of vertices as possible cluster centers are transmitted along the edges of a graph, until a high-quality set of cluster pivots is found; see~\cite{affinity}.
We design the following query-efficient procedure as a baseline:% \textsf{QEAffinity}:
\begin{enumerate}
    \item Pick $k$ random vertices without replacement and query their complete neighborhood.
     Here $k$ is chosen as high as possible within the query budget $Q$, i.e., $$k = \argmax \{ t \mid (2n-t-1)t/2 \le Q \}.$$
    \item Set the affinity of any pair of vertices queried to 1 if there exists an edge.
    \item Set all remaining affinities to zero.
    \item Run the affinity propagation algorithm from~\cite{affinity} on the resulting adjacency matrix.
\end{enumerate}
%is no proper baseline to compare it against. (As discussed in Section~\ref{sec:intro}, the algorithm from~\cite{acc} is essentially the same as $\ouralg$, and appeared six years later.)

We also compare the quality measures for $\ouralg$ and $\ourheur$ for a range of query budgets $Q$ with those from the expected 3-approximation algorithm $\balls$ from~\cite{balls}.
While better approximation factors are possible (2.5 from~\cite{balls}, 2.06 from~\cite{near_opt}), these algorithms require writing a linear program with
$\Omega(n^3)$ constraints and all $\Omega(n^2)$ pairs of vertices need to be queried. By contrast, $\balls$ typically performs much fewer queries, making it more suitable for comparison.

 \begin{figure*}
         \centering
         \begin{tabular}{cccc}
\hspace{-4mm}\includegraphics[width=0.26\textwidth]{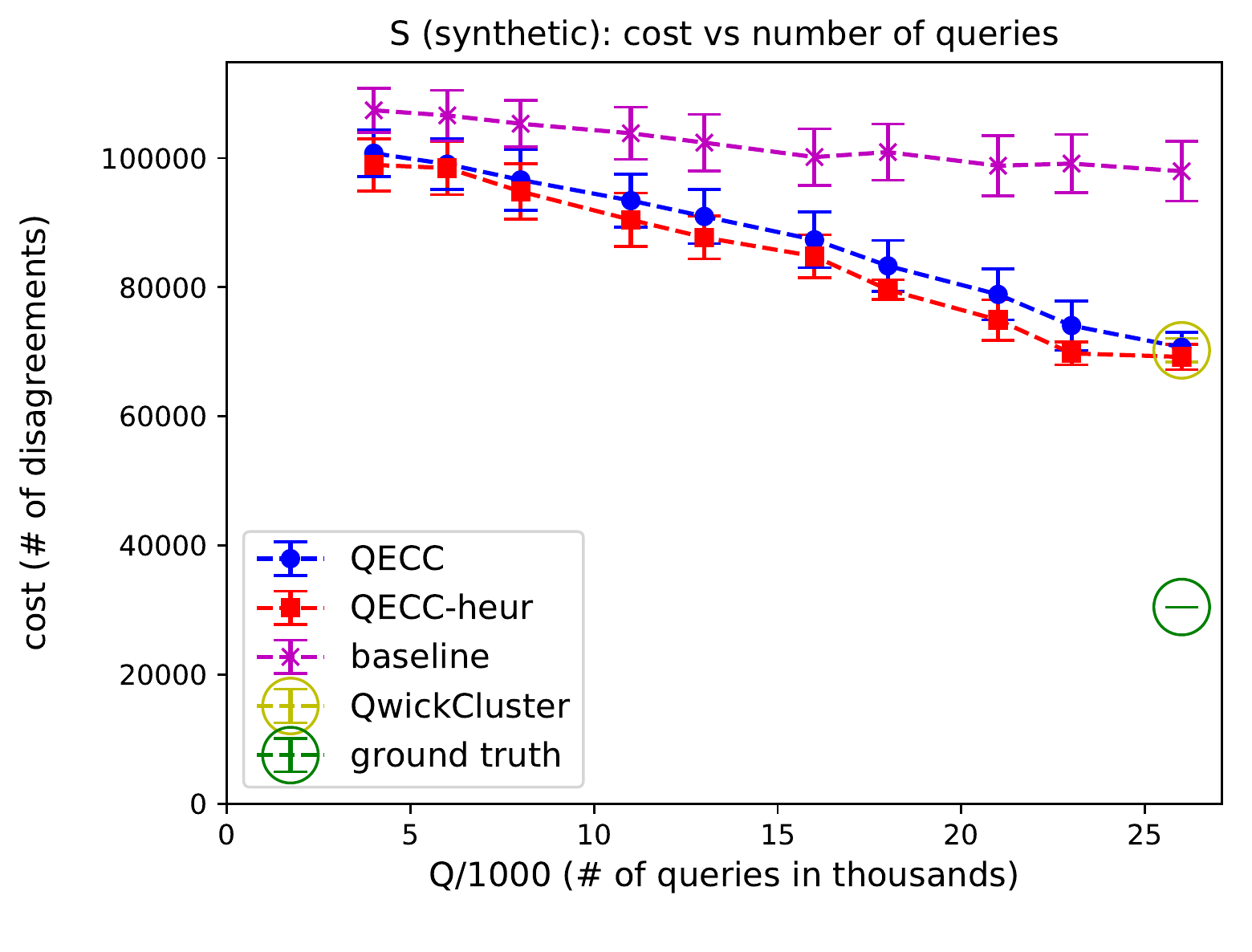} & \hspace{-4mm}\includegraphics[width=0.25\textwidth]{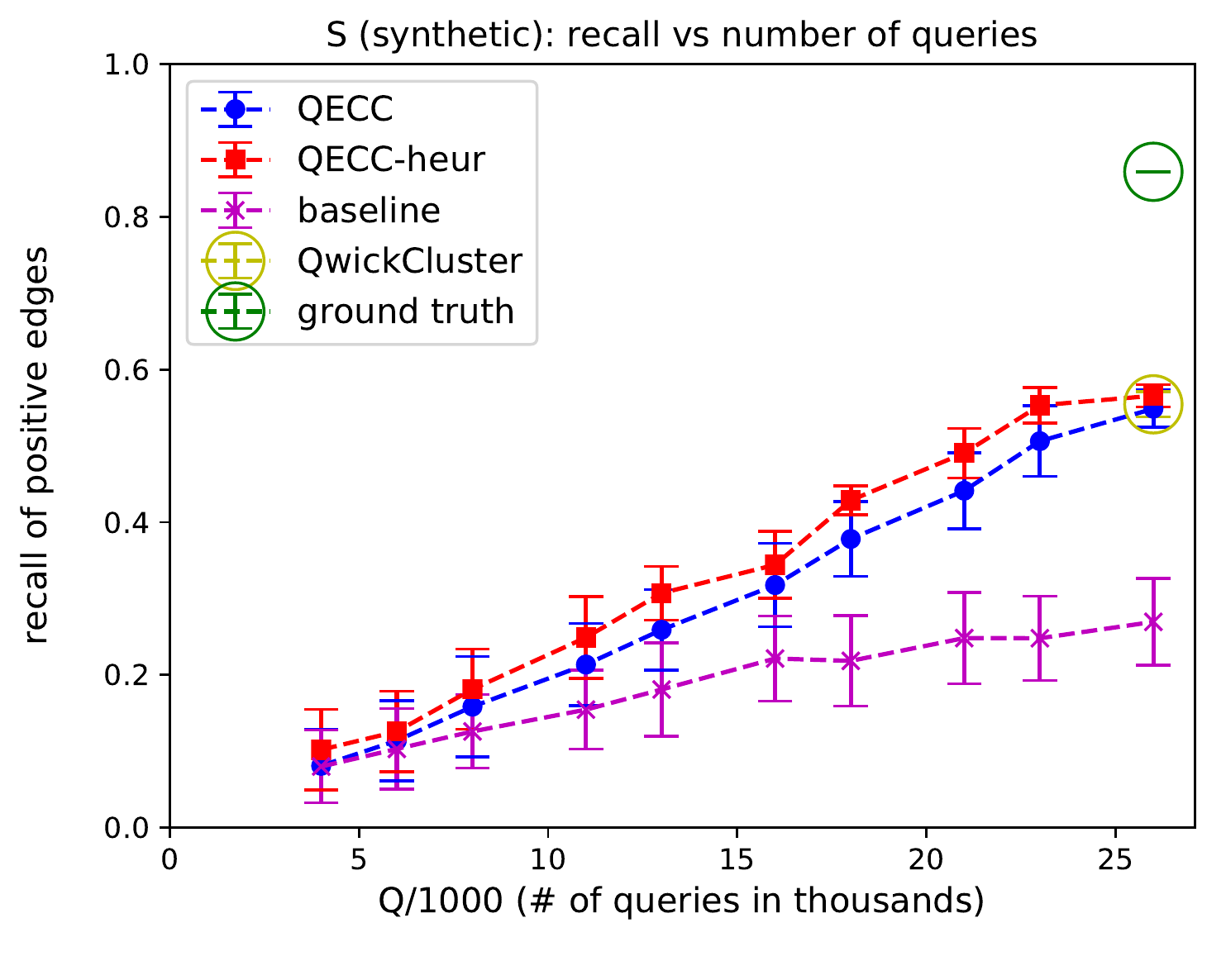} & \hspace{-4mm}\includegraphics[width=0.25\textwidth]{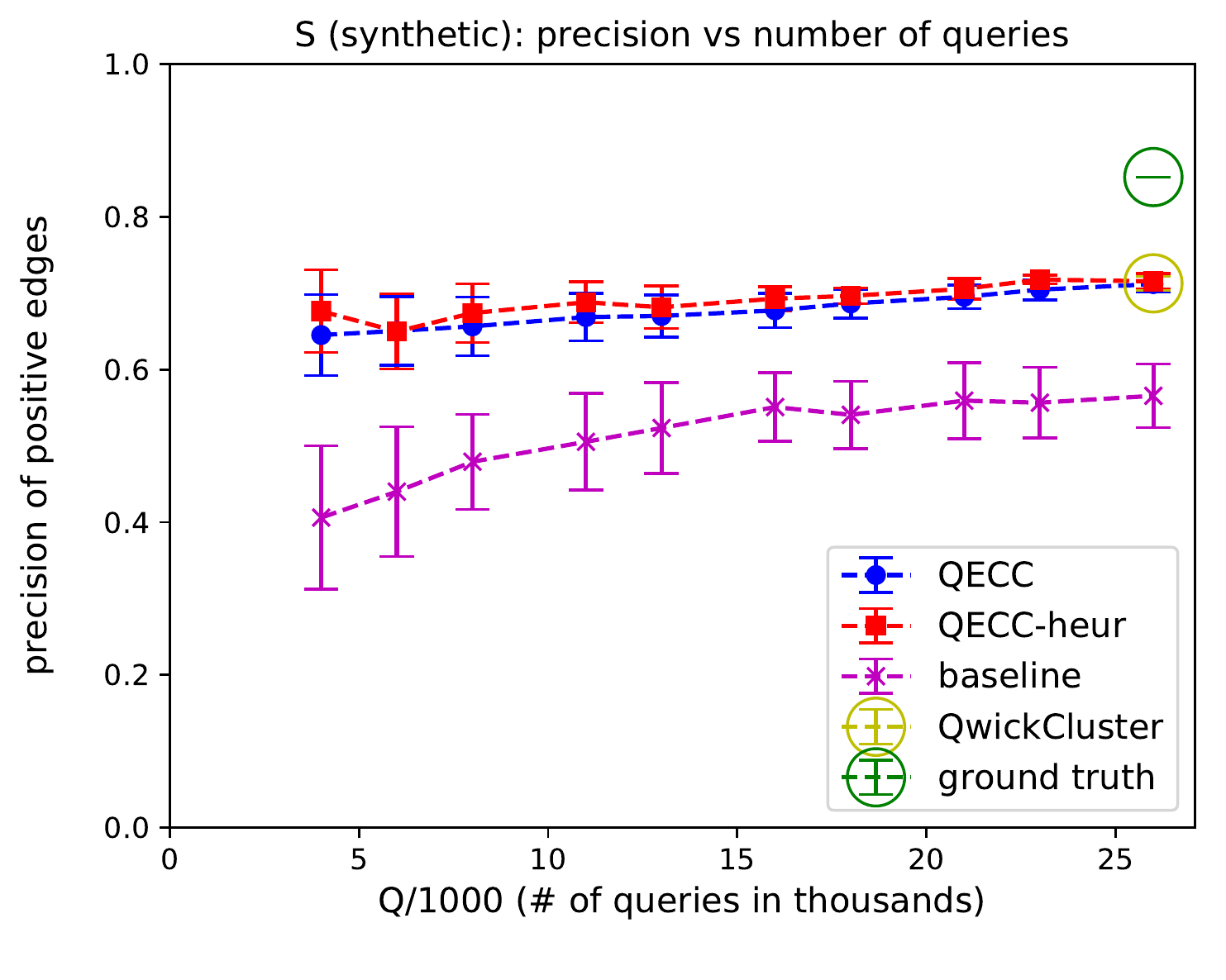} & \hspace{-4mm}\includegraphics[width=0.25\textwidth]{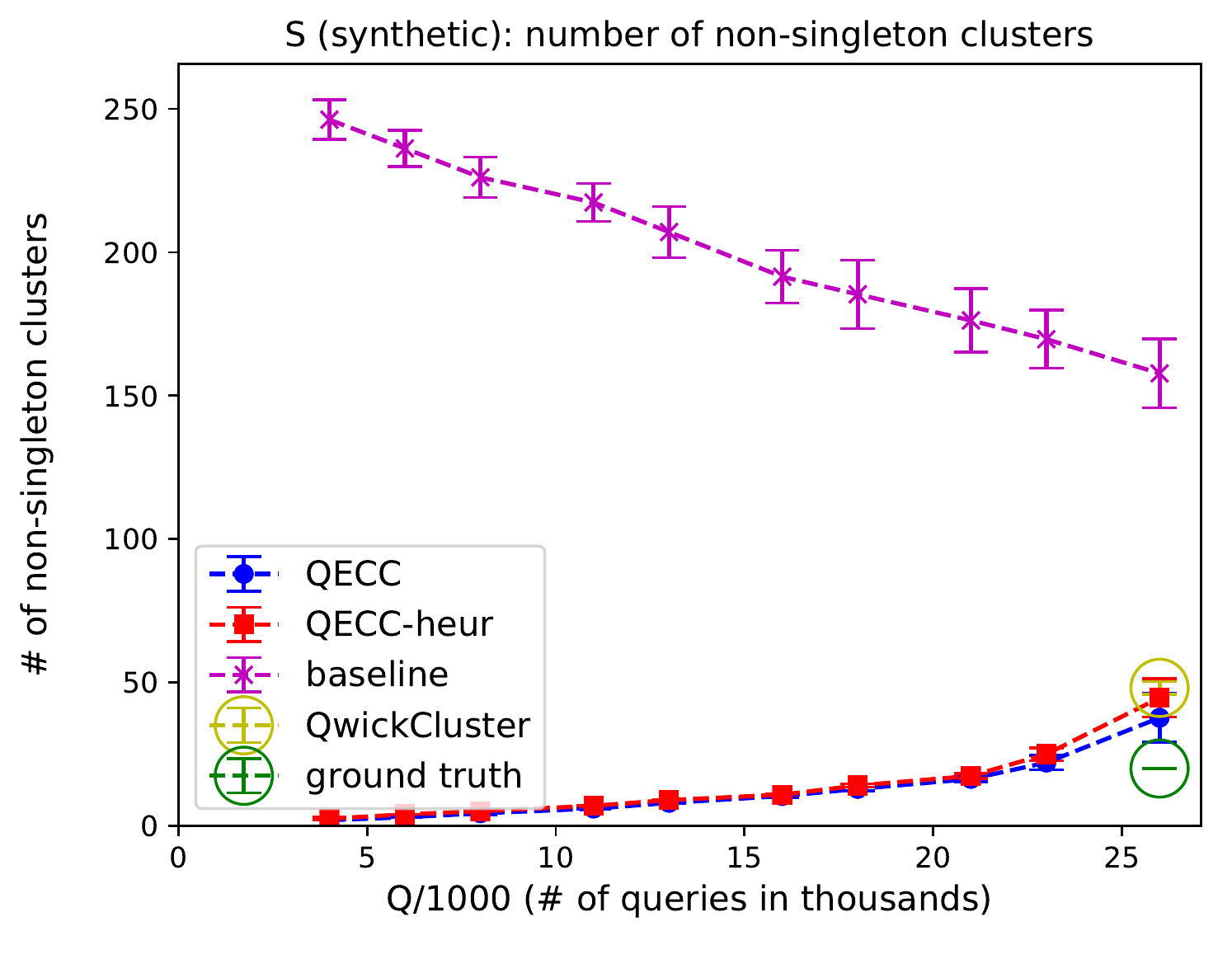} \\
\hspace{-4mm}\includegraphics[width=0.26\textwidth]{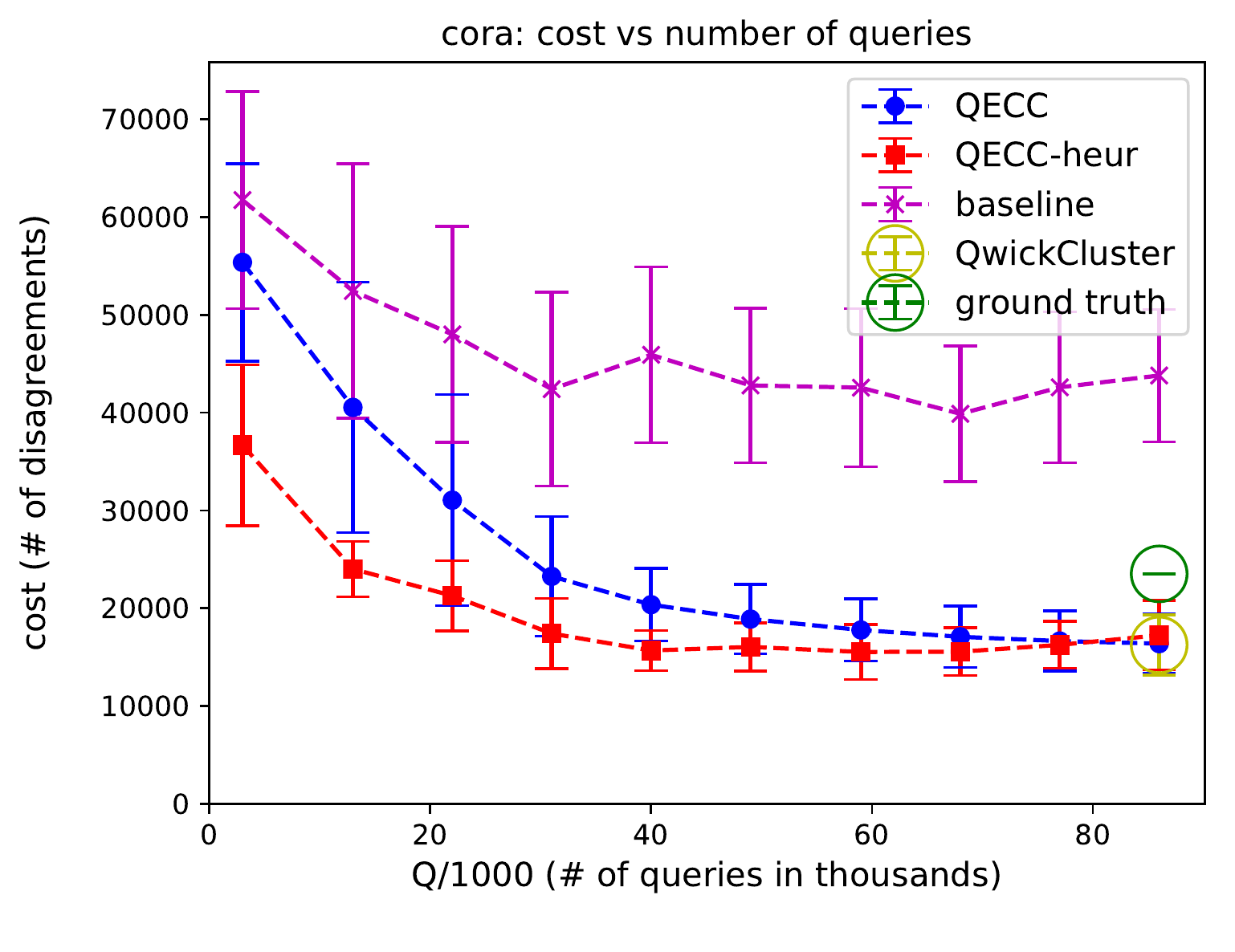} & \hspace{-4mm}\includegraphics[width=0.25\textwidth]{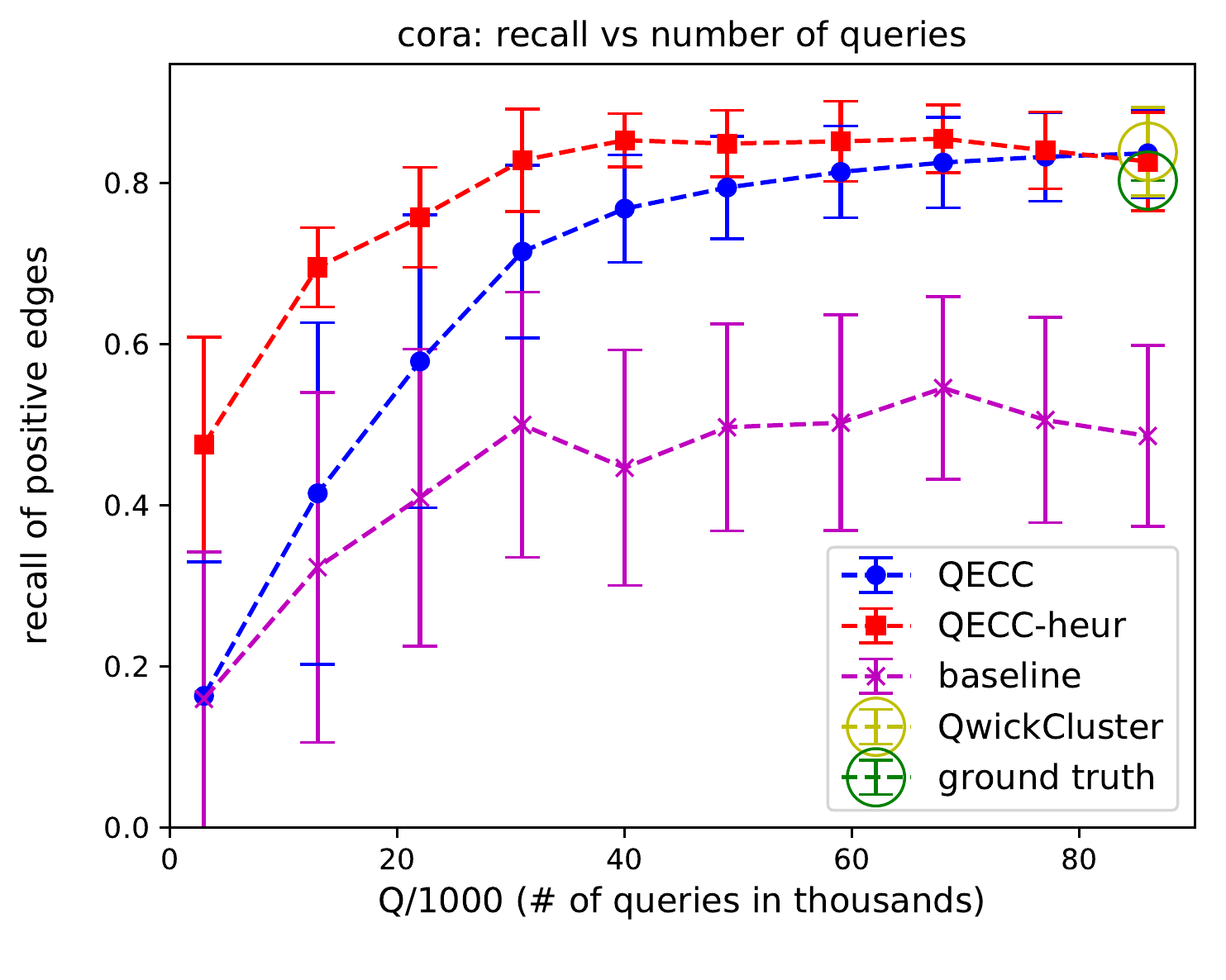} & \hspace{-4mm}\includegraphics[width=0.25\textwidth]{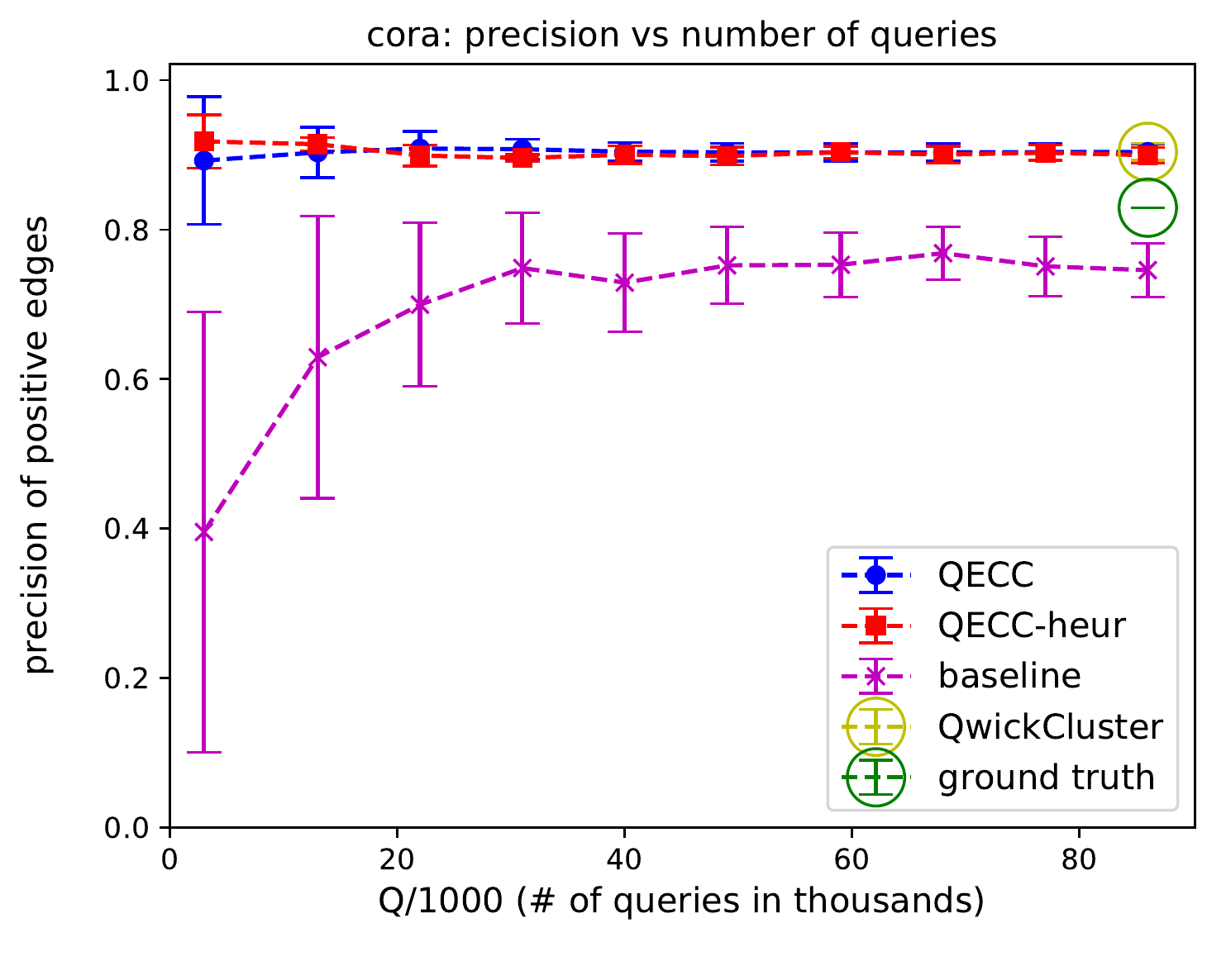} & \hspace{-4mm}\includegraphics[width=0.25\textwidth]{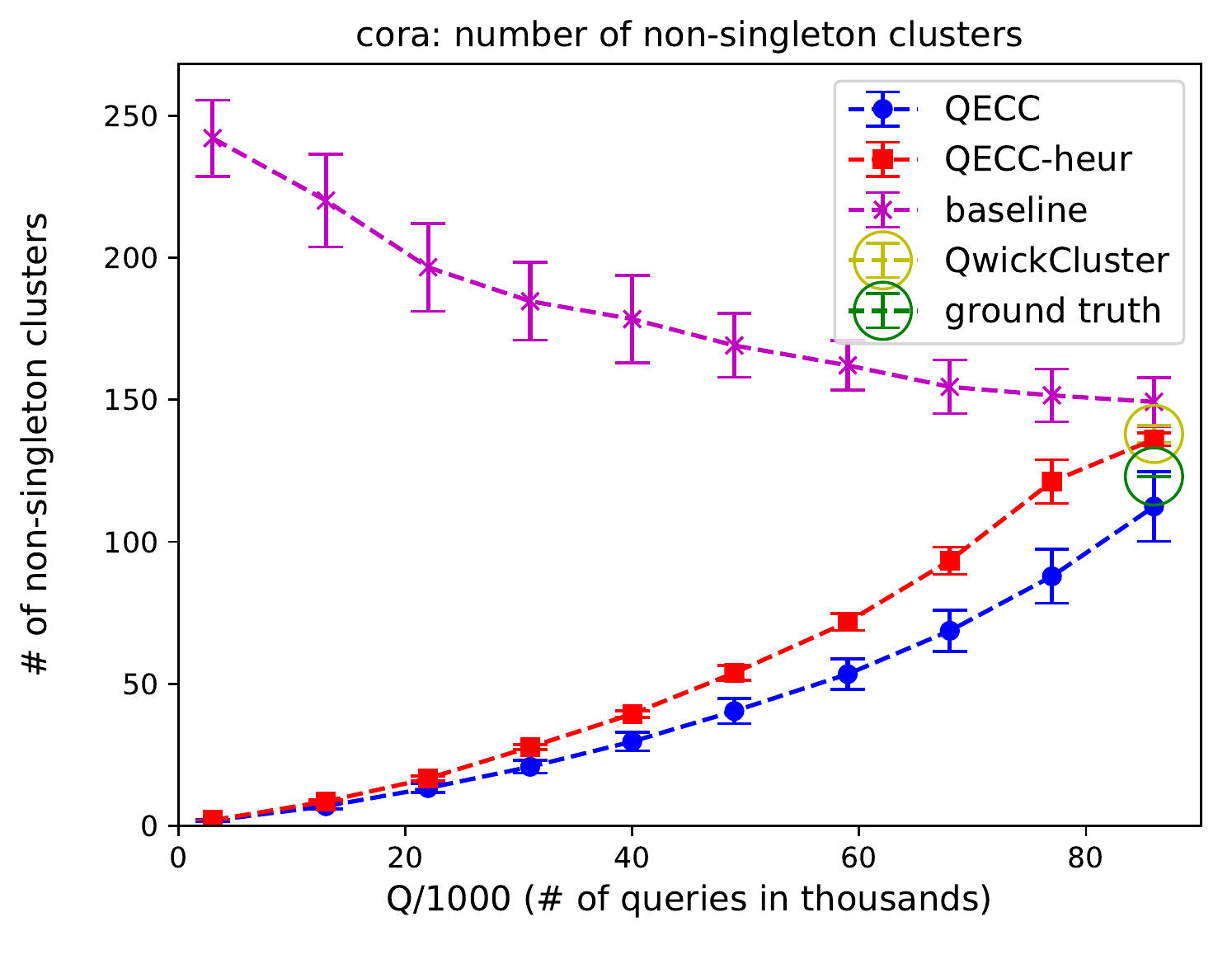} \\
\hspace{-4mm}\includegraphics[width=0.26\textwidth]{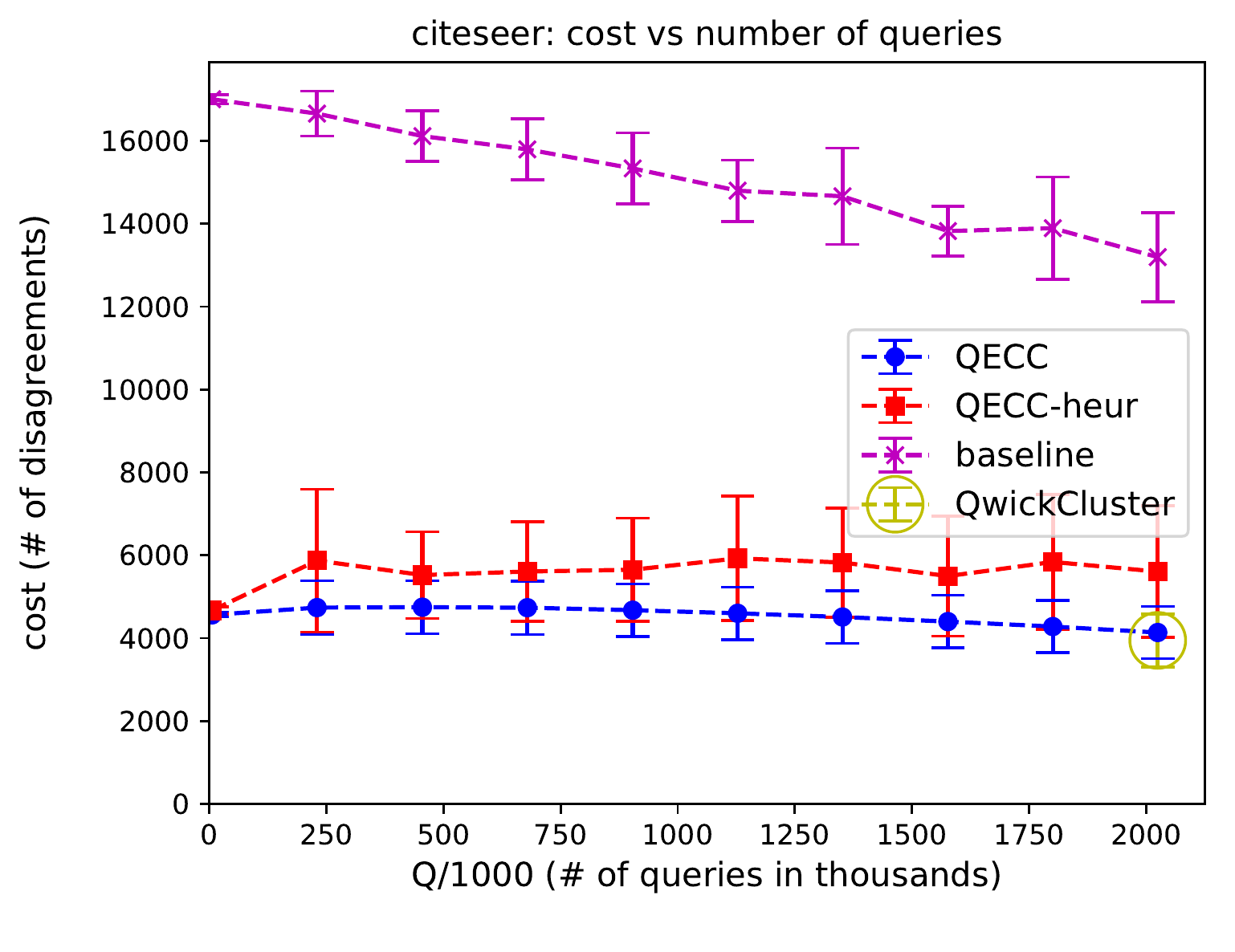} & \hspace{-4mm}\includegraphics[width=0.25\textwidth]{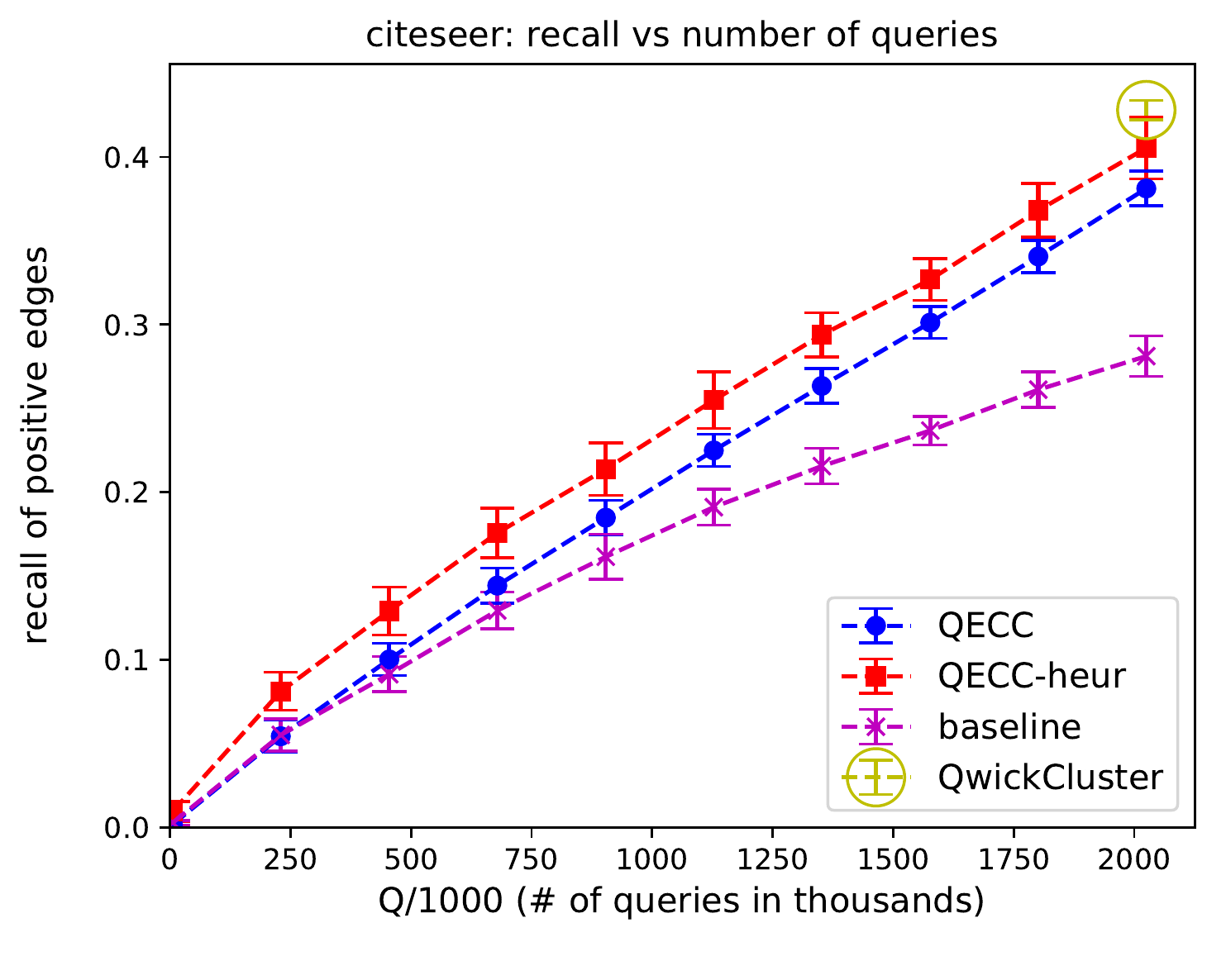} &
\hspace{-4mm}\includegraphics[width=0.25\textwidth]{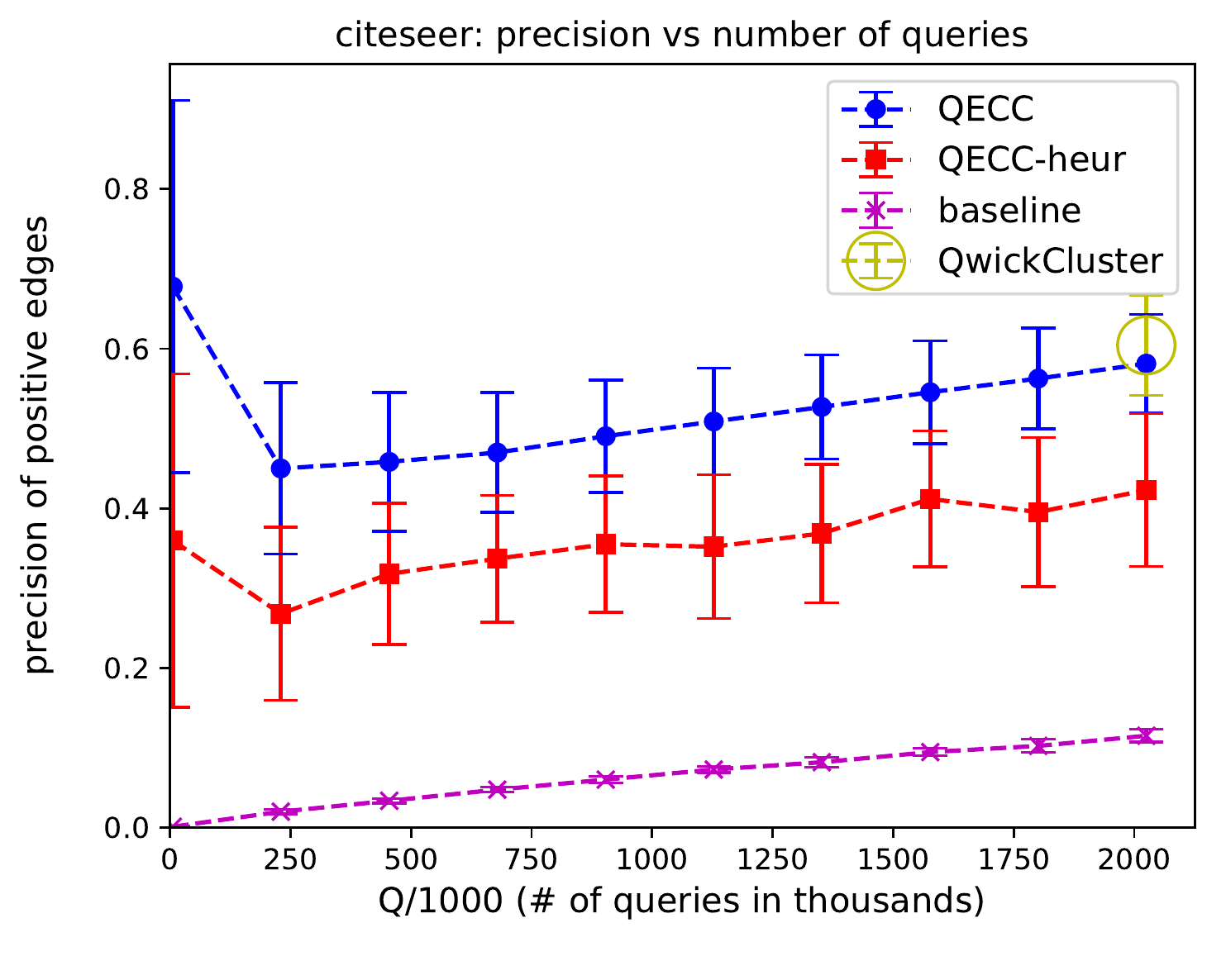} & \hspace{-4mm}\includegraphics[width=0.26\textwidth]{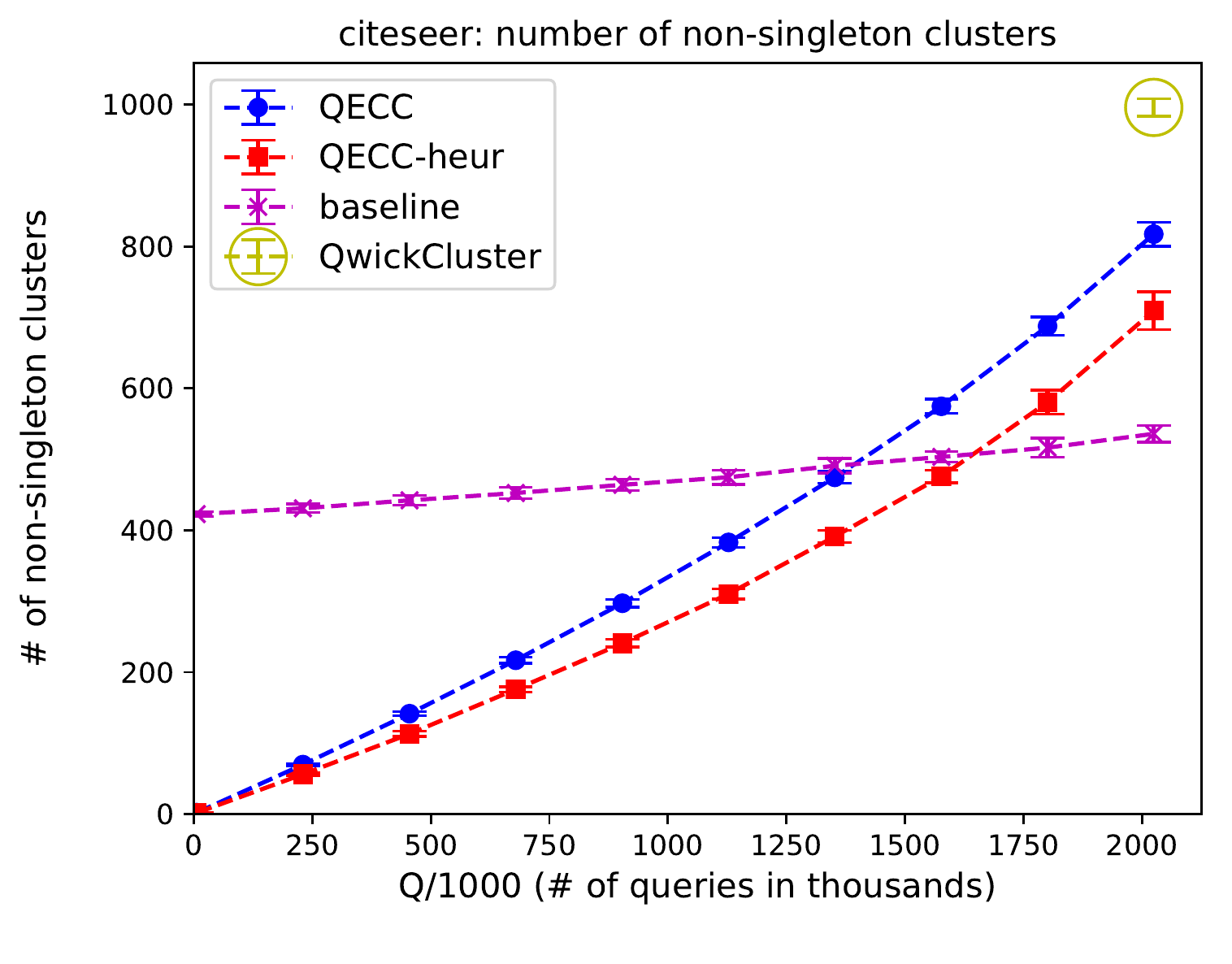}\\
\hspace{-4mm}\includegraphics[width=0.25\textwidth]{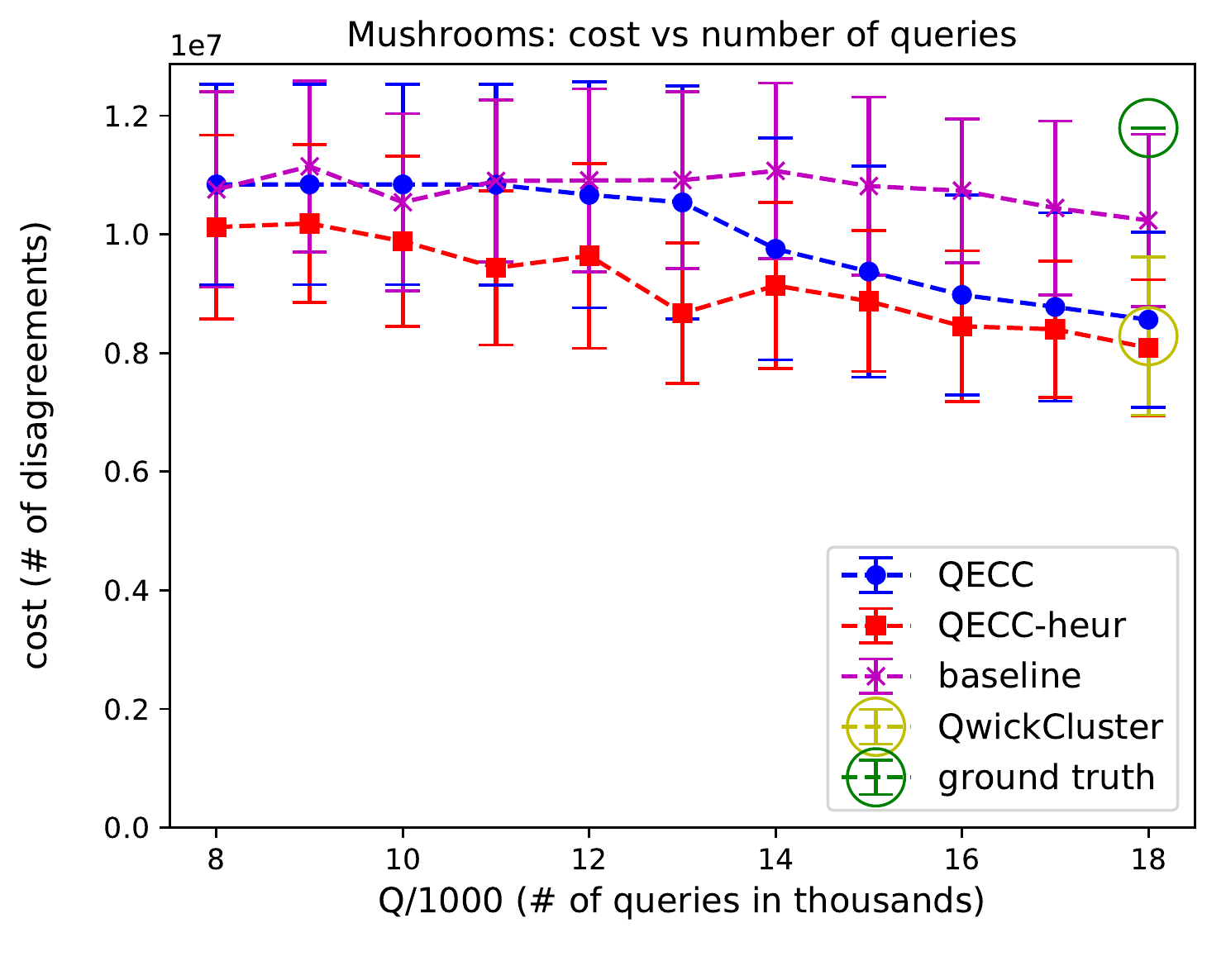} & \hspace{-4mm}\includegraphics[width=0.25\textwidth]{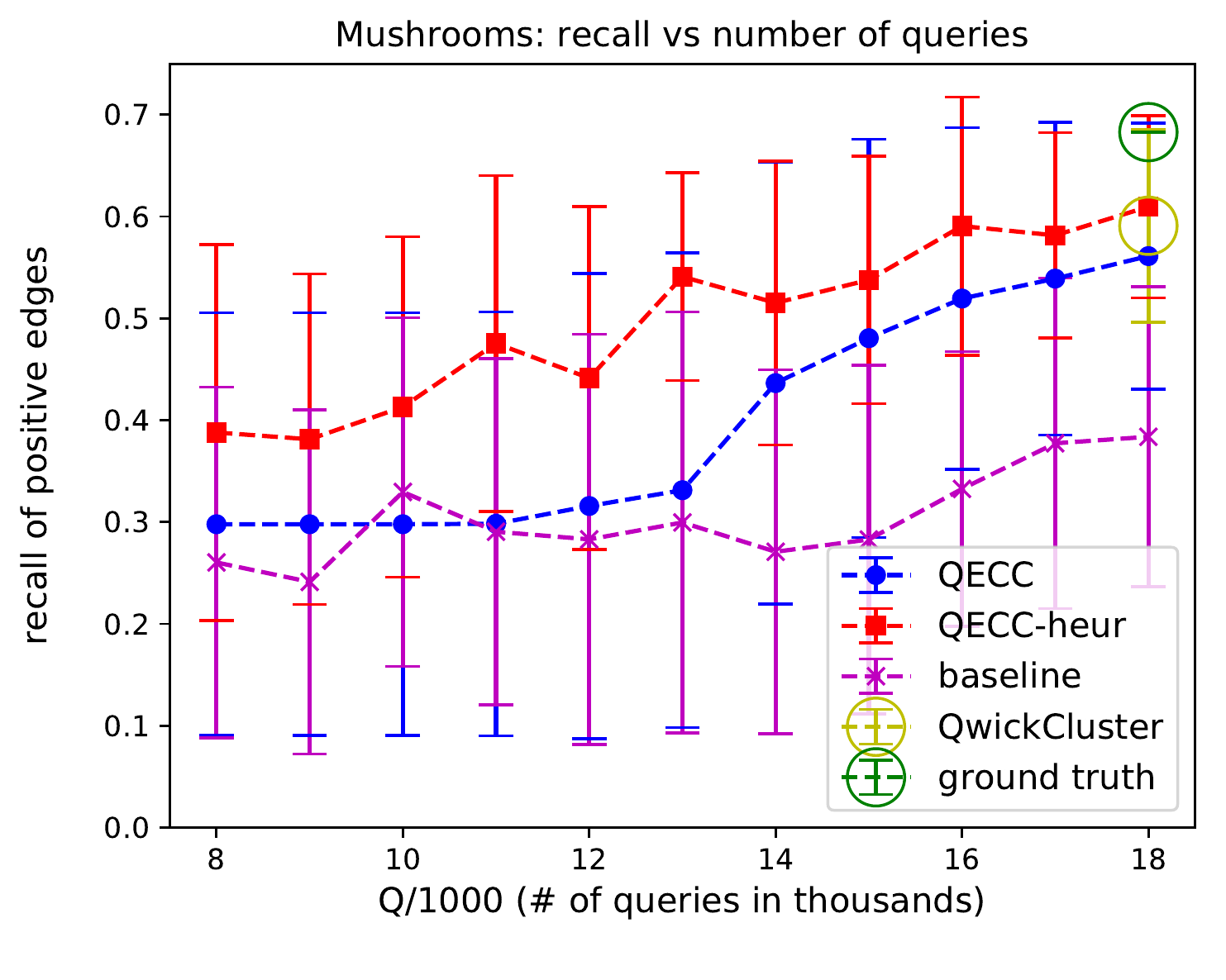} & \hspace{-4mm}\includegraphics[width=0.25\textwidth]{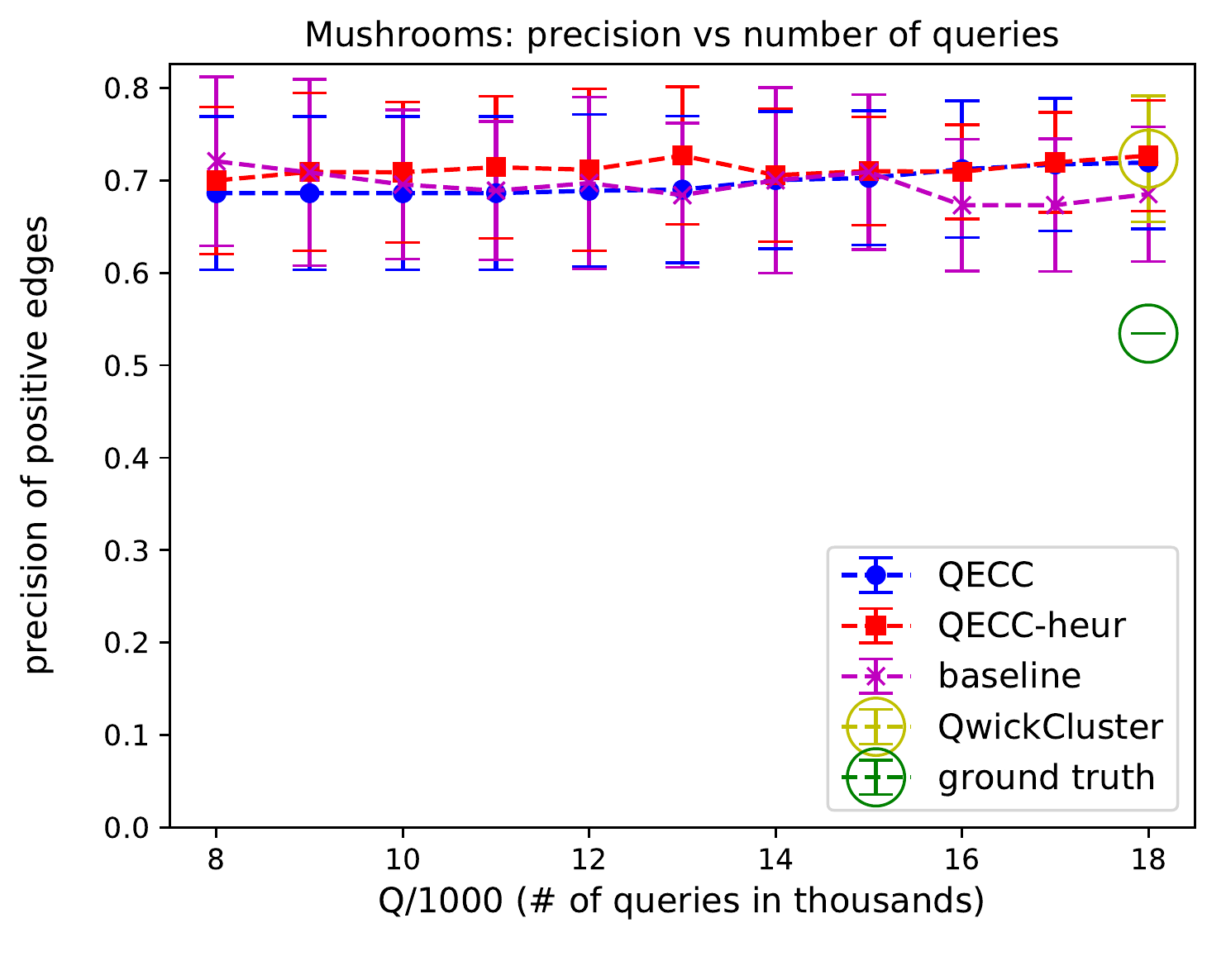} & \hspace{-4mm}\includegraphics[width=0.25\textwidth]{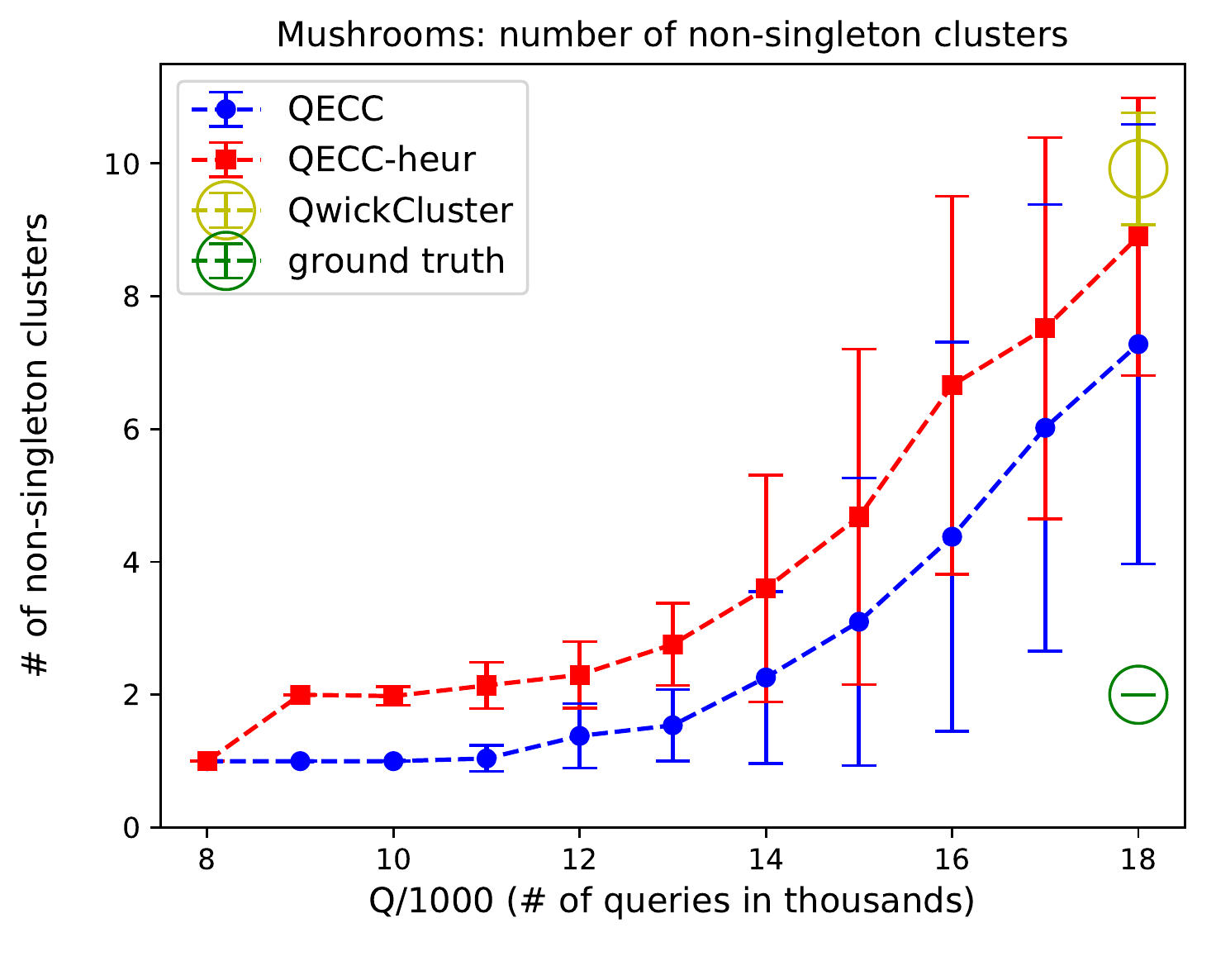}
  \end{tabular}
\caption{\label{fig:findings} Accuracy measures of our two algorithms, the baseline, and ground truth on the
    datasets of Table~\ref{tab:datasets}.
}
\end{figure*}

\spara{Synthetic graphs.}
We construct a family of synthetic graphs $\calS = \{S(n, k, \alpha, \beta)\}$, parameterized by the number of vertices $n$, number of clusters in the ground truth $k$, imbalance
factor $\alpha$, and noise rate $\beta$.
    The ground truth $T(n, k, \alpha)$ for $S(n, k, \alpha, \beta)$ consists of one clique of size $\alpha n / k$ and $k-1$ cliques of size $(1-\alpha) n / k$, all disjoint.
    To construct the input graph $S(n,k,\alpha,\beta)$, we flip the sign of every edge between same-cluster nodes in $T(n,k,\alpha)$ with probability $\beta$, and we flip the sign
    of every edge between distinct-cluster nodes with
    probability $\beta / (k-1)$. (This ensures that the number total number of positive and negative edges flipped is roughly the same.)

\mycomment{
    We construct to kinds of synthetic graphs, parameterized by the number of vertices $n$, number of clusters in the ground truth $k$, and imbalance factor $\alpha$:
    \begin{itemize}
    \item The ground truth in the first kind ($S_1$) consists of one cluster of size $\alpha n / k$ and $k-1$ clusters of size $(1-\alpha) n / k$.
    \item The ground truth in the second kind ($S_2$) consists of roughly $k = O(\log n/\alpha)$ clusters and is constructed as follows:
    we put each vertex in the first cluster independently with probability $\alpha$, then remove the cluster and proceed recursively.
    \end{itemize}
    Note that the $\alpha$ parameters are not comparable across the two kinds of synthetic graphs.

    To construct the input graph, we employ a noise model parameterized by an internal edge error ratio $\beta$ and an external edge error ratio $\gamma$. For every pair of nodes in
    the same ground-truth cluster, we make it a negative edge with probability $\beta$; and for every pair of nodes in different ground-truth clusters, we make it a positive edge with
    probability $\beta$.
}

\spara{Real-world graphs.}
For our experiments with real-world graphs, we choose three with very different characteristics:
\begin{itemize}
\item The $\textsc{cora}$ dataset\footnote{\url{https://github.com/sanjayss34/corr-clust-query-esa2019}}, where each node is a scientific
publication represented by a string determined by its
title, authors, venue, and date. Following~\cite{same_cluster}, nodes are joined by a positive edge when the Jaro string similarity between them exceeds or equals $0.5$.
\item The $\textsc{Citeseer}$ dataset\footnote{\url{https://github.com/kimiyoung/planetoid}},
a record of publication citations for Alchemy. We put an edge between two publications if one of them cites the other~\cite{embeddings}.
\item The $\textsc{Mushrooms}$ dataset\footnote{\url{https://archive.ics.uci.edu/ml/datasets/mushroom}}, including descriptions of mushrooms classified as either edible or
poisonous, corresponding to the two ground-truth clusters. Each mushroom is described by a set of features. To construct the graph, we remove the edible/poisonous feature and place an edge
between two mushrooms if they differ on at most half the remaining features. This construction has been inspired by~\cite{clustering_aggregation}, who show that high-quality clusterings can often
be obtained by aggregating clusterings based on single features.

%It consists of 1.9K nodes, 191 clusters with the largest cluster-size being 236.
%No noise is added to the input graphs since the optimal solution has non-zero cost for them.
%\item The $\textsc{newsgroups}$ dataset, comprising around 18000 newsgroups posts on 20 topics.
     %We extract TF-IDF vectors from unigram tokens, and use $1000$ queries to determinate an estimate $\tau$ to the average cosine similarity between pairs of posts. After $\tau$ has been found, we consider that there is a positive edge between two posts if the cosine similarity is equal to or larger than $\tau$.
\end{itemize}

 \begin{figure*}
         \centering
         \begin{tabular}{ccc}
\hspace{-4mm}\includegraphics[width=0.32\textwidth]{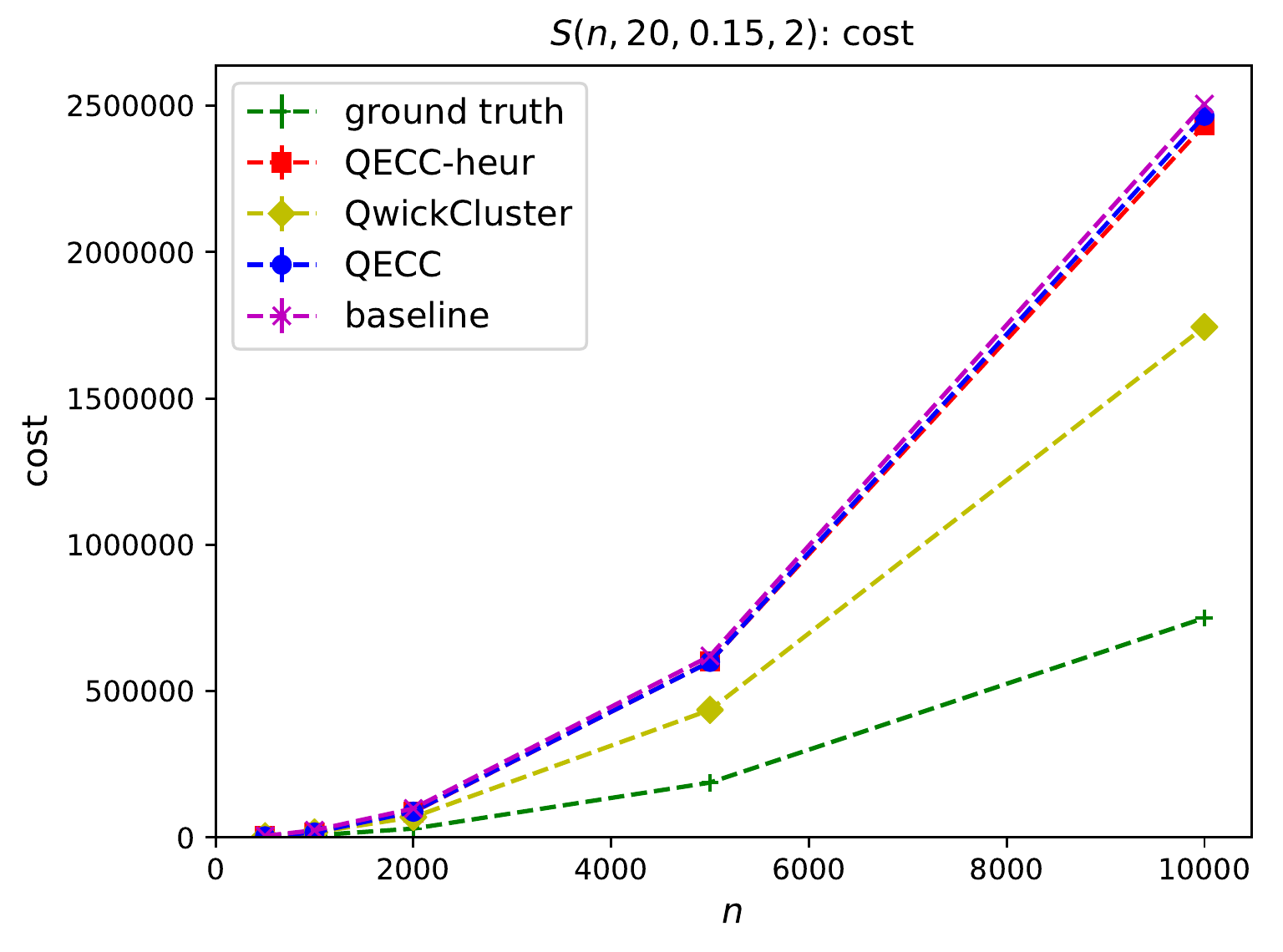} & \includegraphics[width=0.3\textwidth]{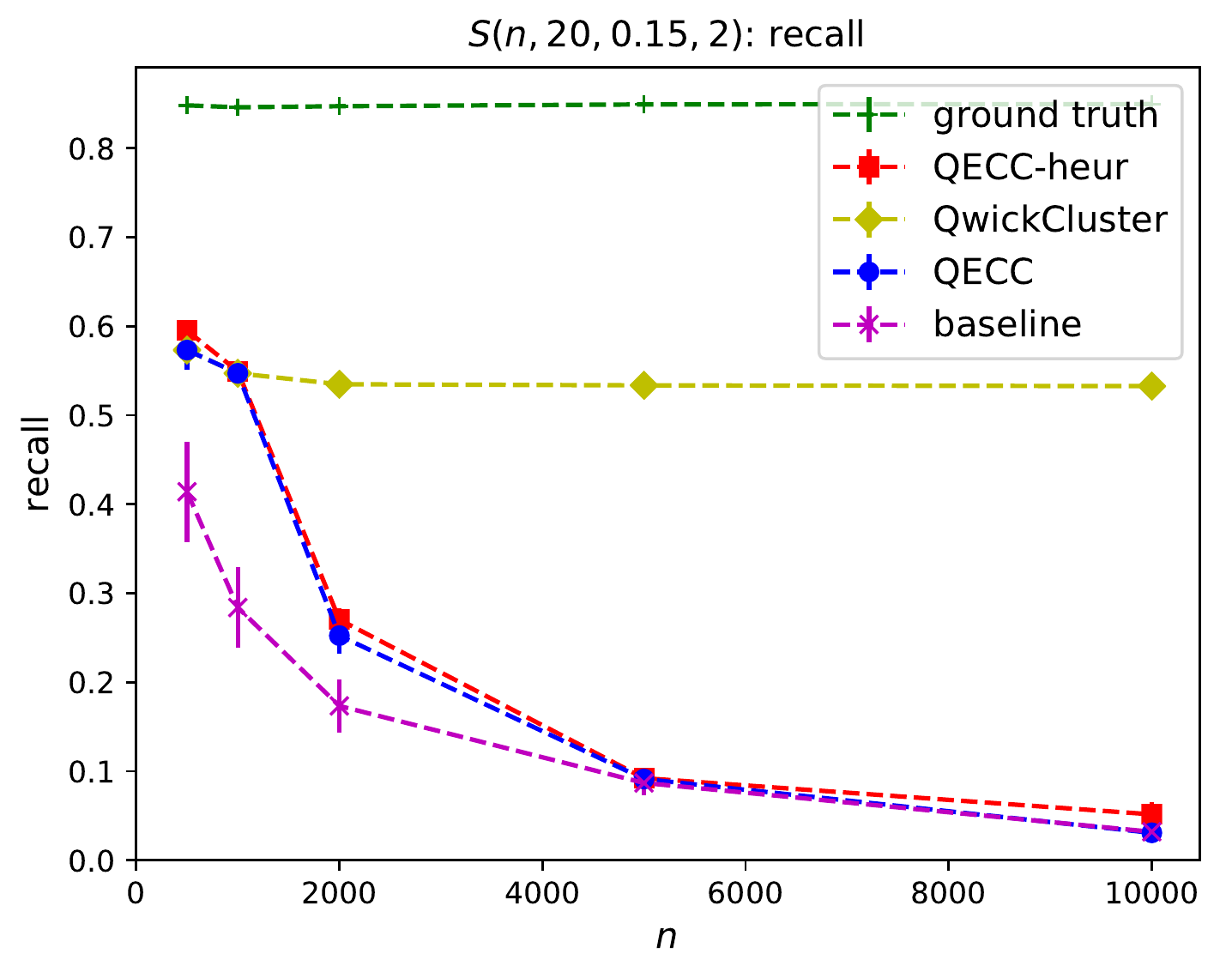} &
\includegraphics[width=0.3\textwidth]{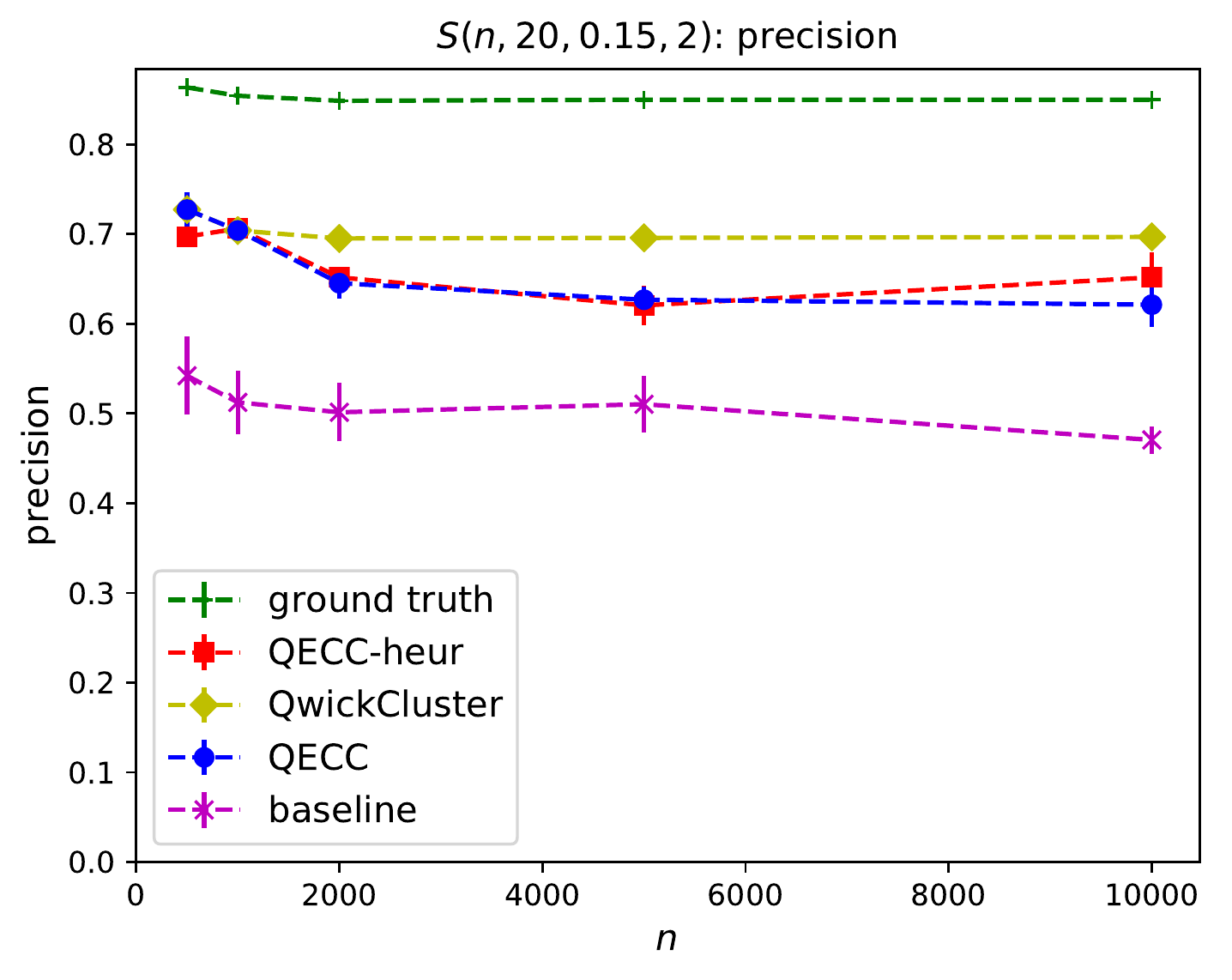} \\
\\
\hspace{-4mm}\includegraphics[width=0.32\textwidth]{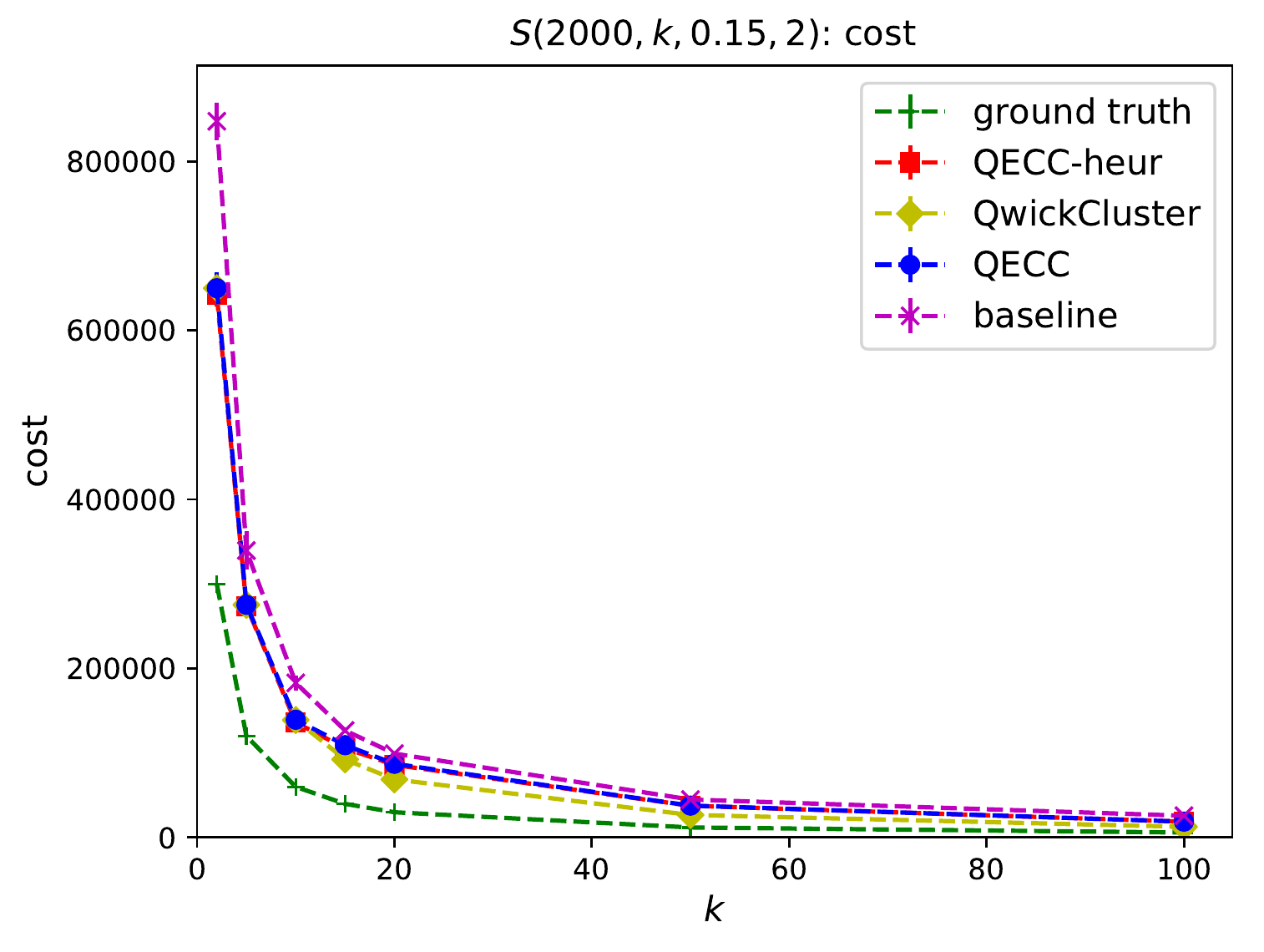} & \includegraphics[width=0.3\textwidth]{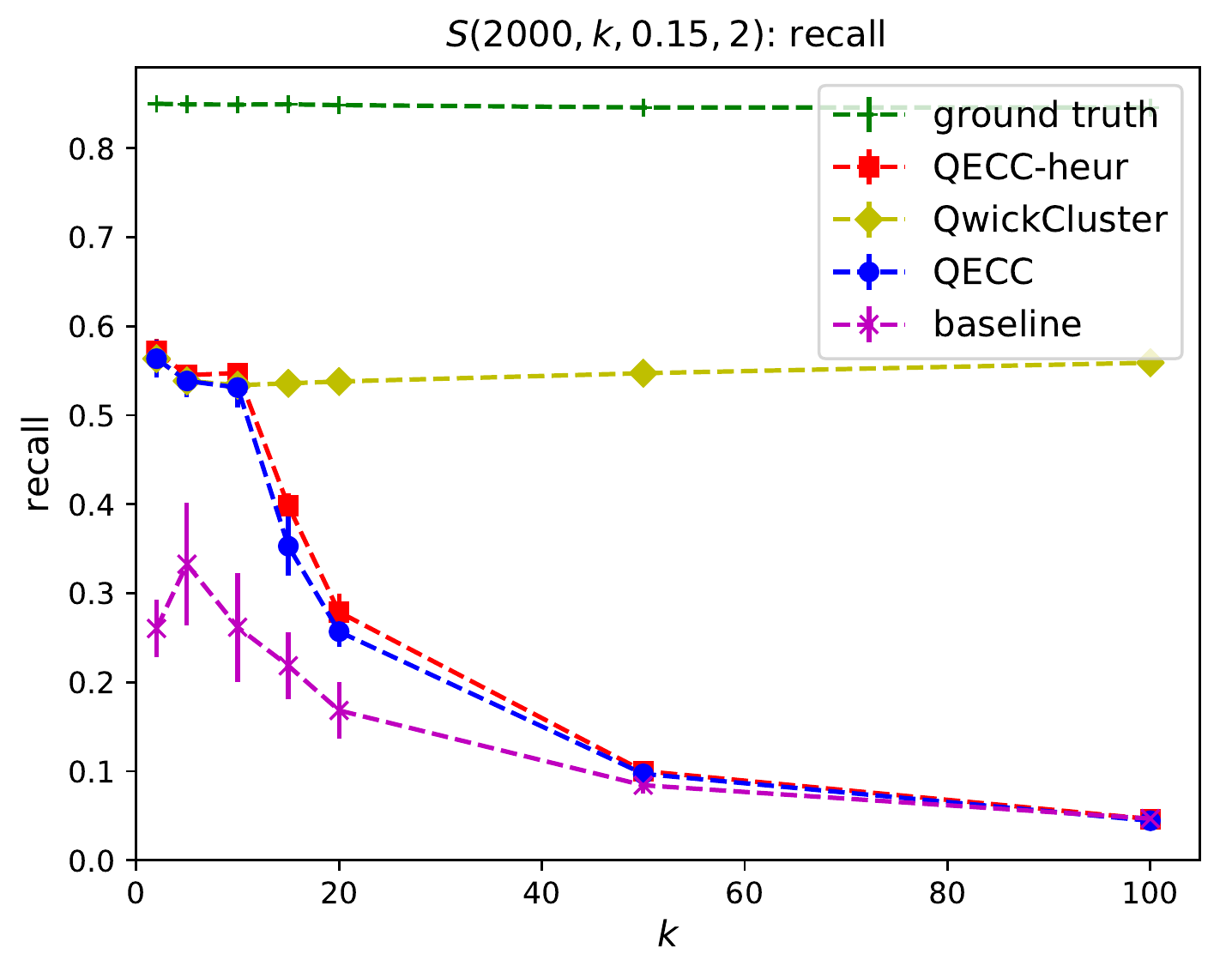} &
\includegraphics[width=0.3\textwidth]{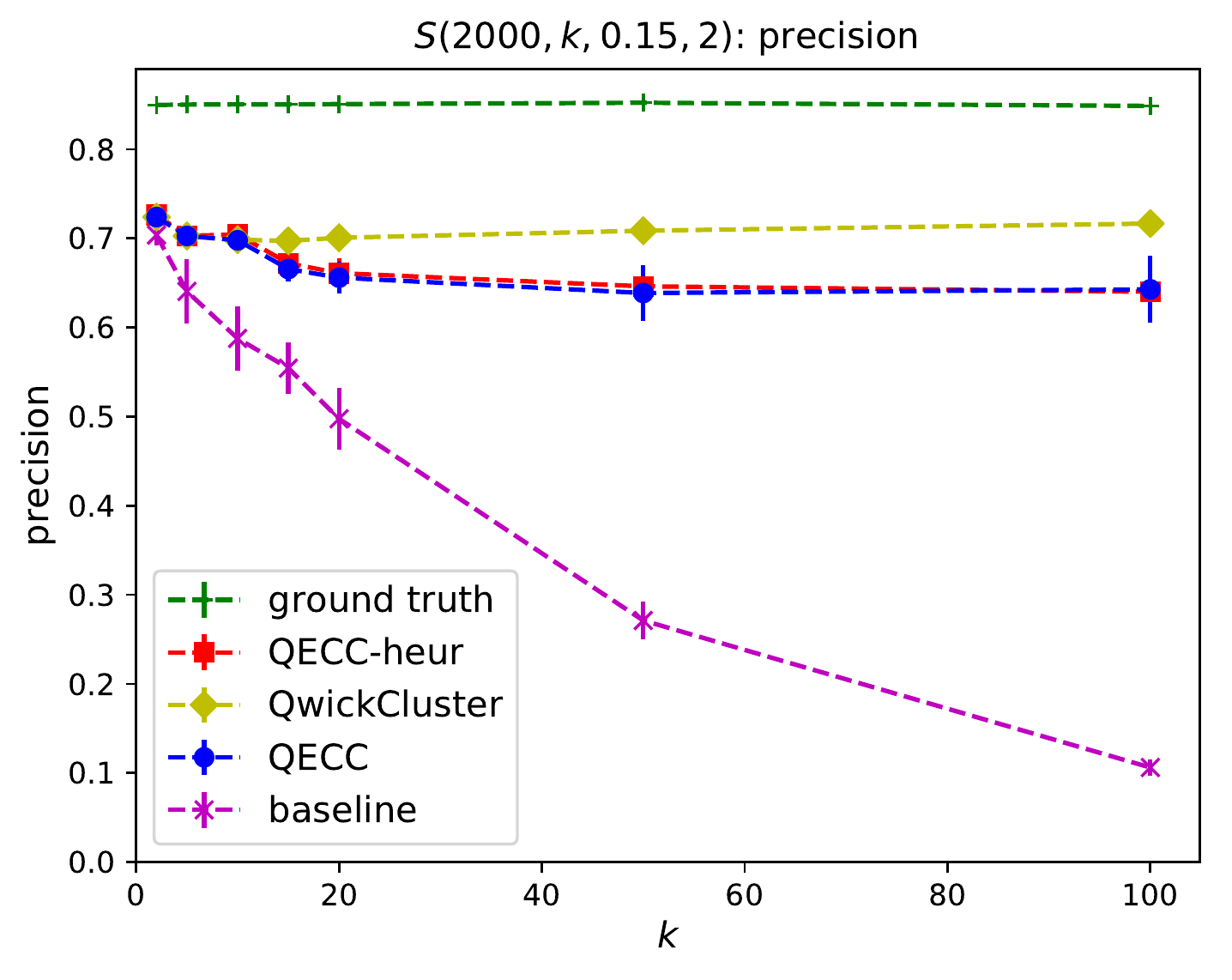} \\
\hspace{-4mm}\includegraphics[width=0.32\textwidth]{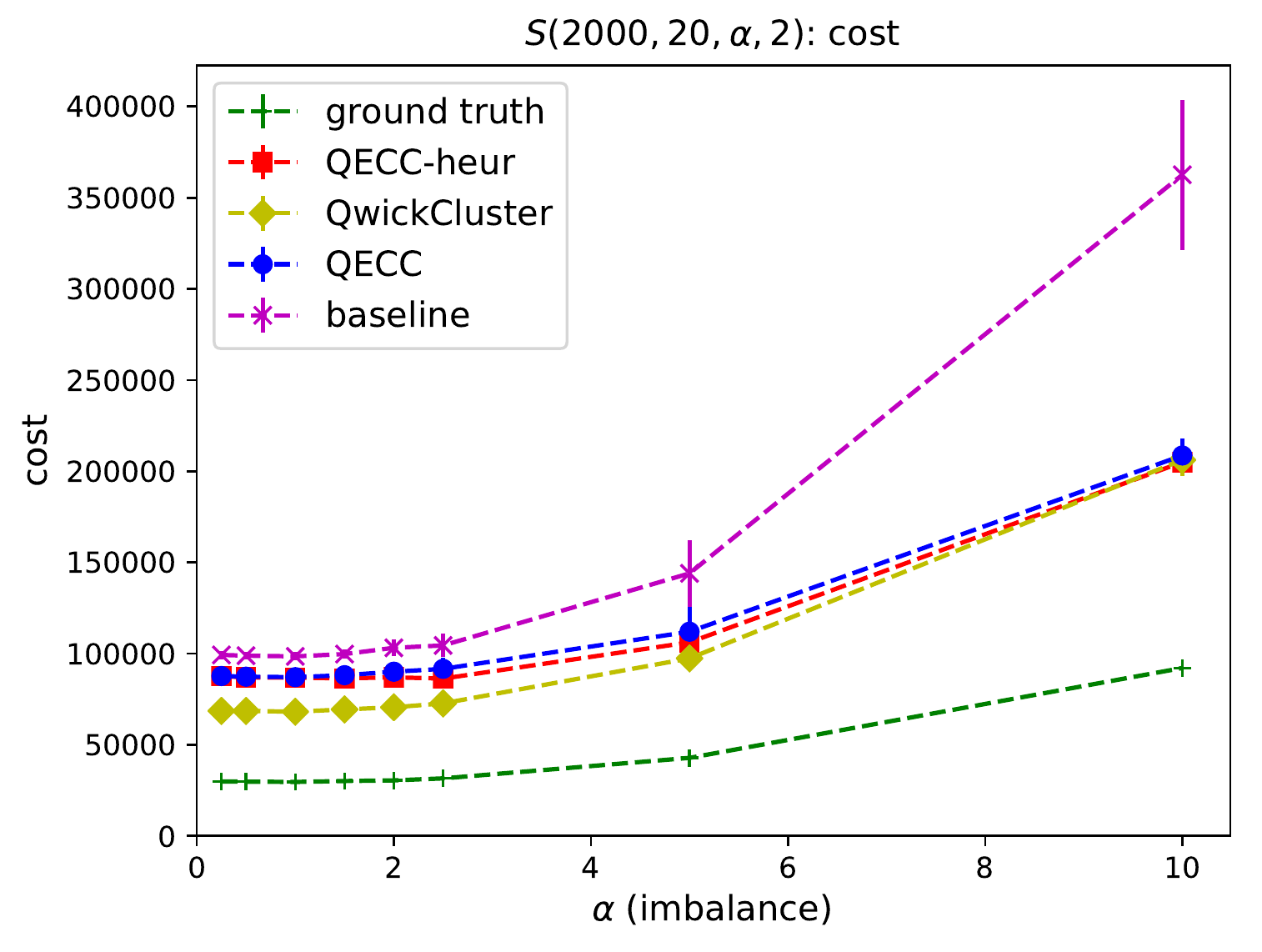} & \includegraphics[width=0.3\textwidth]{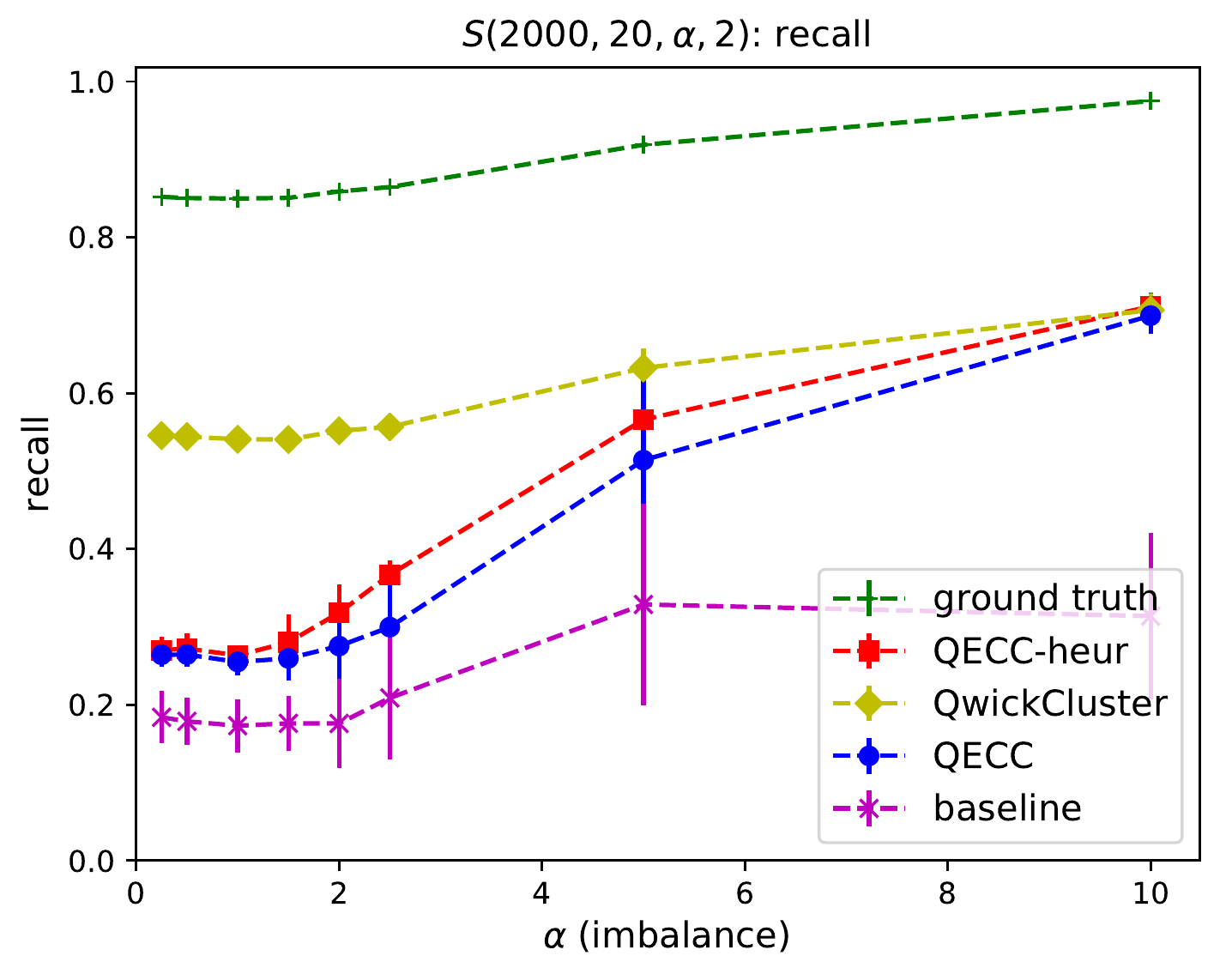}&
\includegraphics[width=0.3\textwidth]{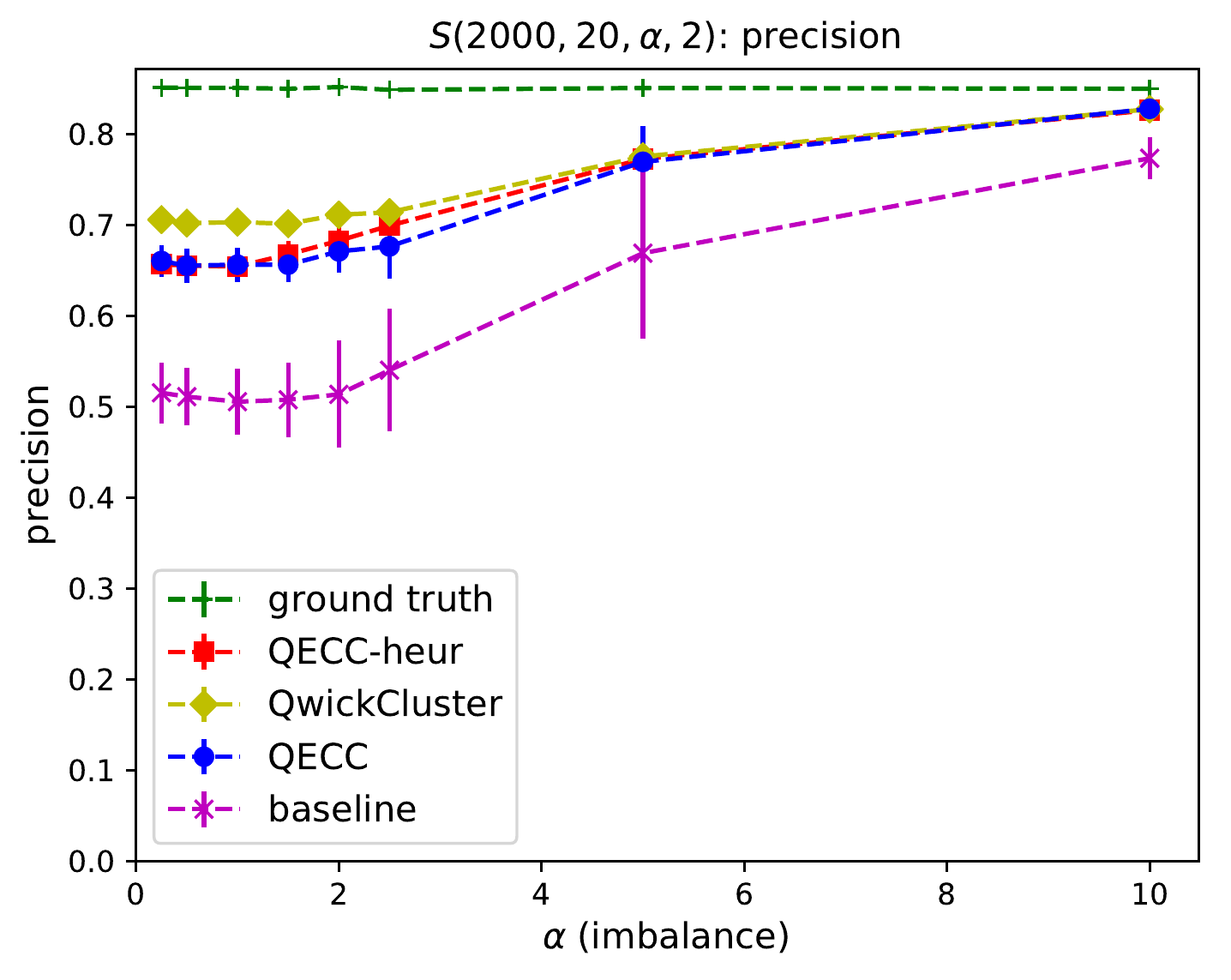}  \\
\hspace{-4mm}\includegraphics[width=0.32\textwidth]{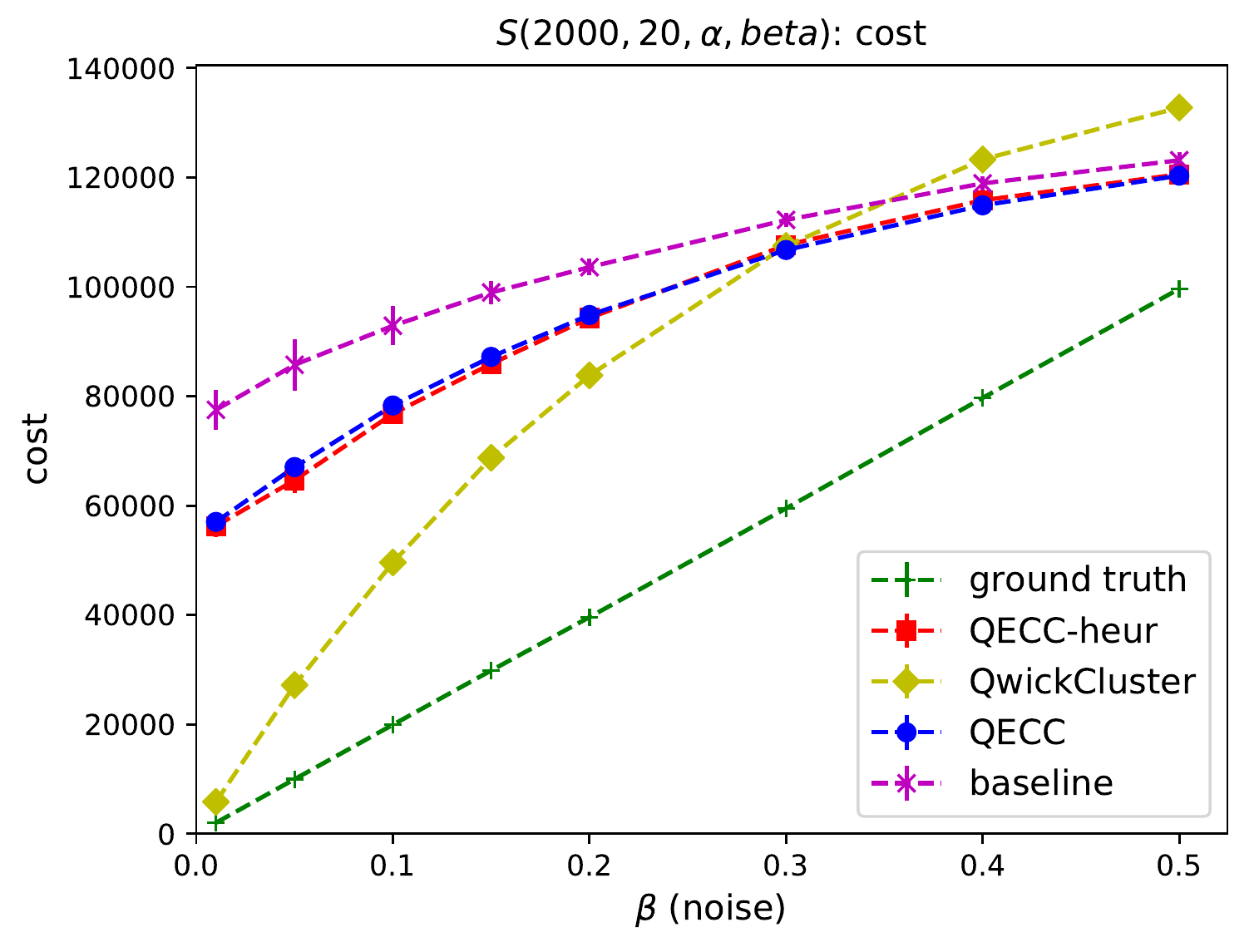} & \includegraphics[width=0.3\textwidth]{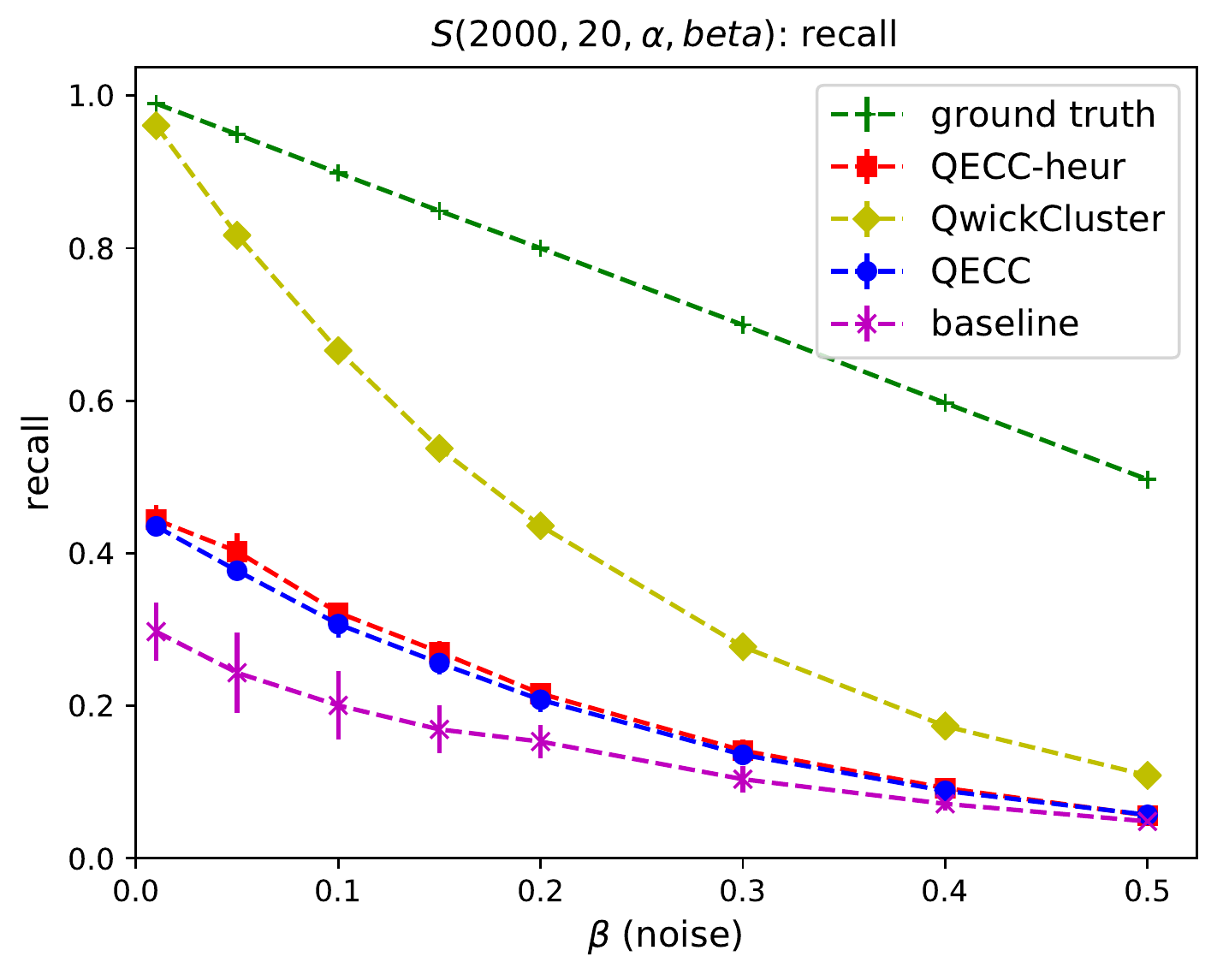}&
\includegraphics[width=0.3\textwidth]{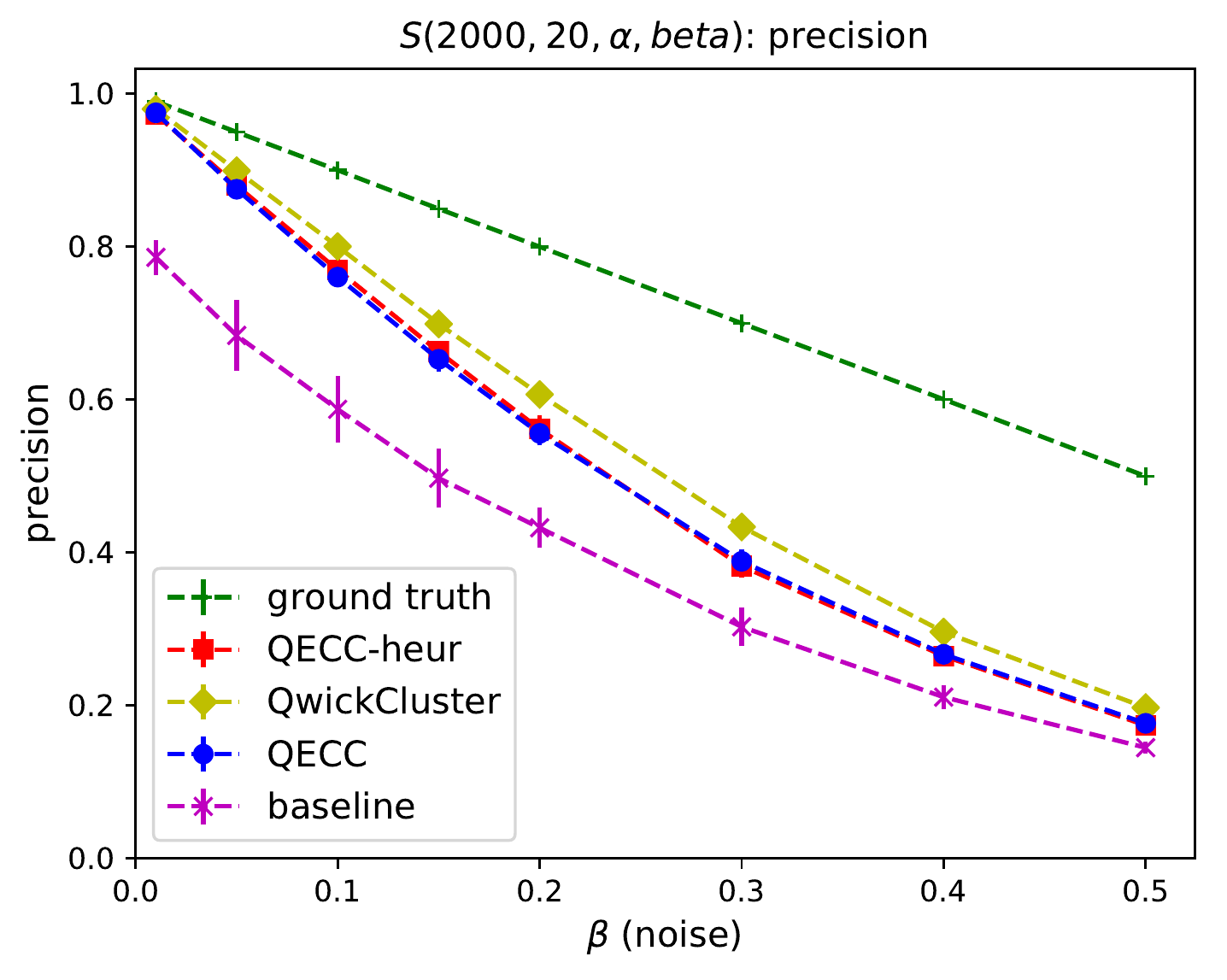}  \\
  \end{tabular}
\caption{\label{fig:variation}
    Effect of graph size $n$ (1rst row), number of clusters $k$, (2nd row), imbalance parameter $\alpha$ (3rd row) and noise parameter $\beta$ (4rth row)
    total cost and recall, for a fixed number $Q=15000$ of queries, except for $\balls$.
}
\end{figure*}

\spara{Methodology.}
All the algorithms we test are randomized, hence we run each of them $50$ times and compute the empirical average and standard deviations of the total cost, precision and
recall values. We compute the average number $A$ of queries made by $\balls$ and then run our algorithm with an allowance of queries ranging from $2n$ to $A$ at regular intervals.

We use synthetic graphs to study how cost and recall vary in terms of (1) number of nodes $n$; (2) number of clusters $k$; (3) imbalance parameter $\alpha$; (4) noise parameter $\beta$.
For each plot, we fix all remaining parameters and vary one of them.

As the  runtime for $\ouralg$ scales linearly with the number of queries $Q$, which is an input parameter, we chose not to
report detailed runtimes. We note that a simple Python implementation of our methods runs in under two seconds in all cases on an Intel i7 CPU at 3.7Ghz, and
runs faster than the affinity propagation baseline we used (as implemented in Scikit-learn).%, while a C++ implementation runs in less than $0.1$ seconds.

\subsection{Experimental results}

 \begin{figure*}[h!]
         \centering
         \begin{tabular}{ccc}
\includegraphics[width=0.32\textwidth]{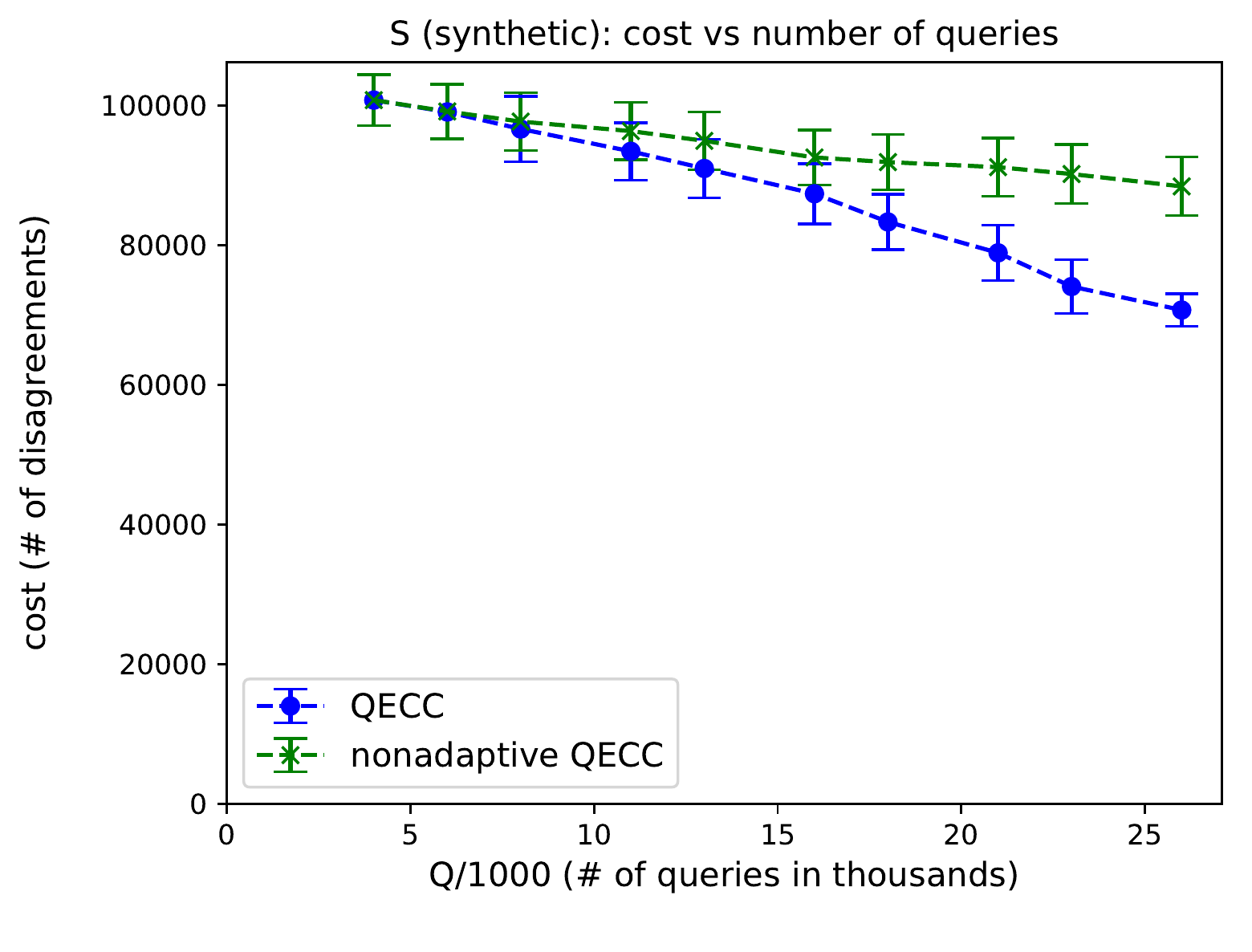} & \includegraphics[width=0.31\textwidth]{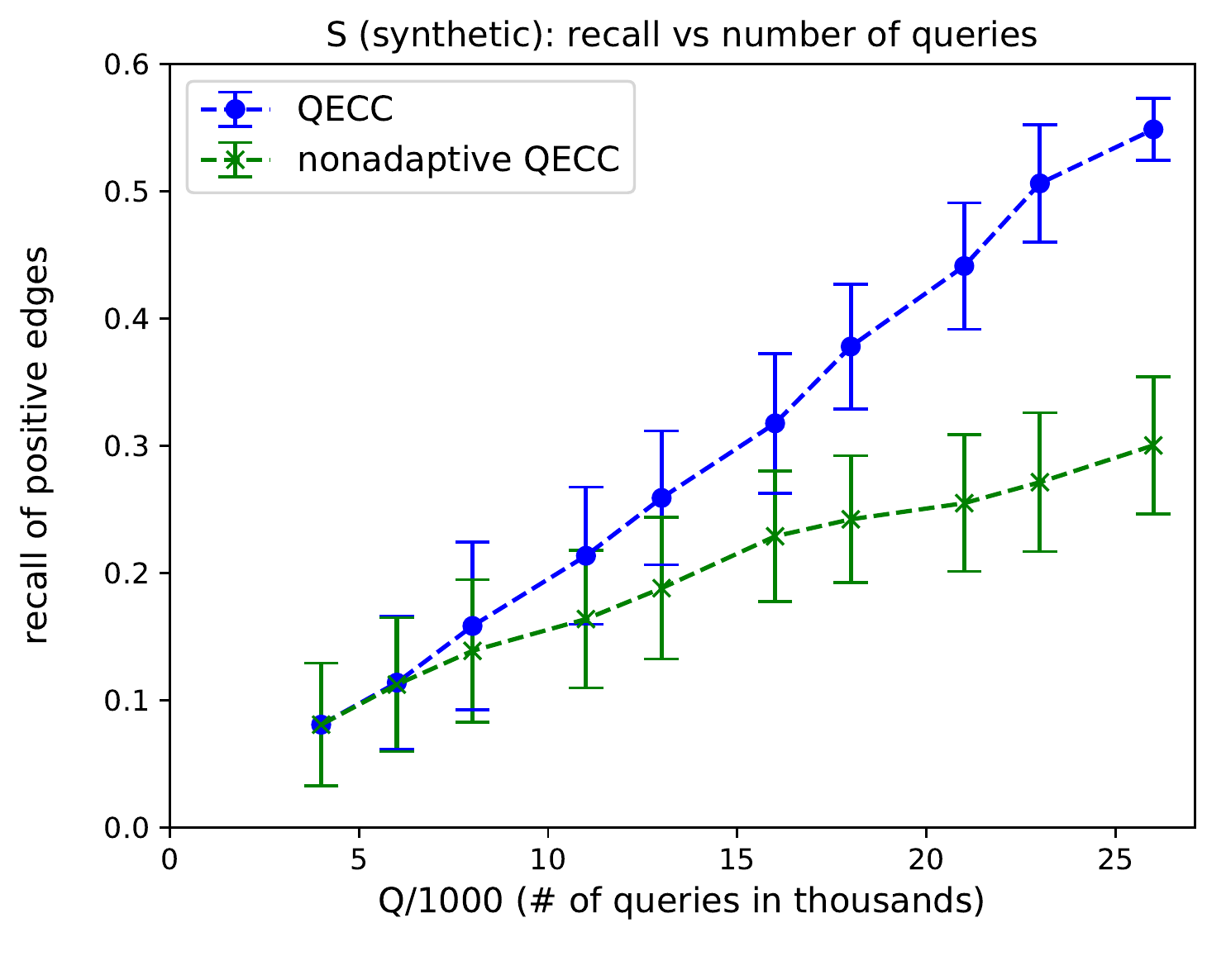} & \includegraphics[width=0.31\textwidth]{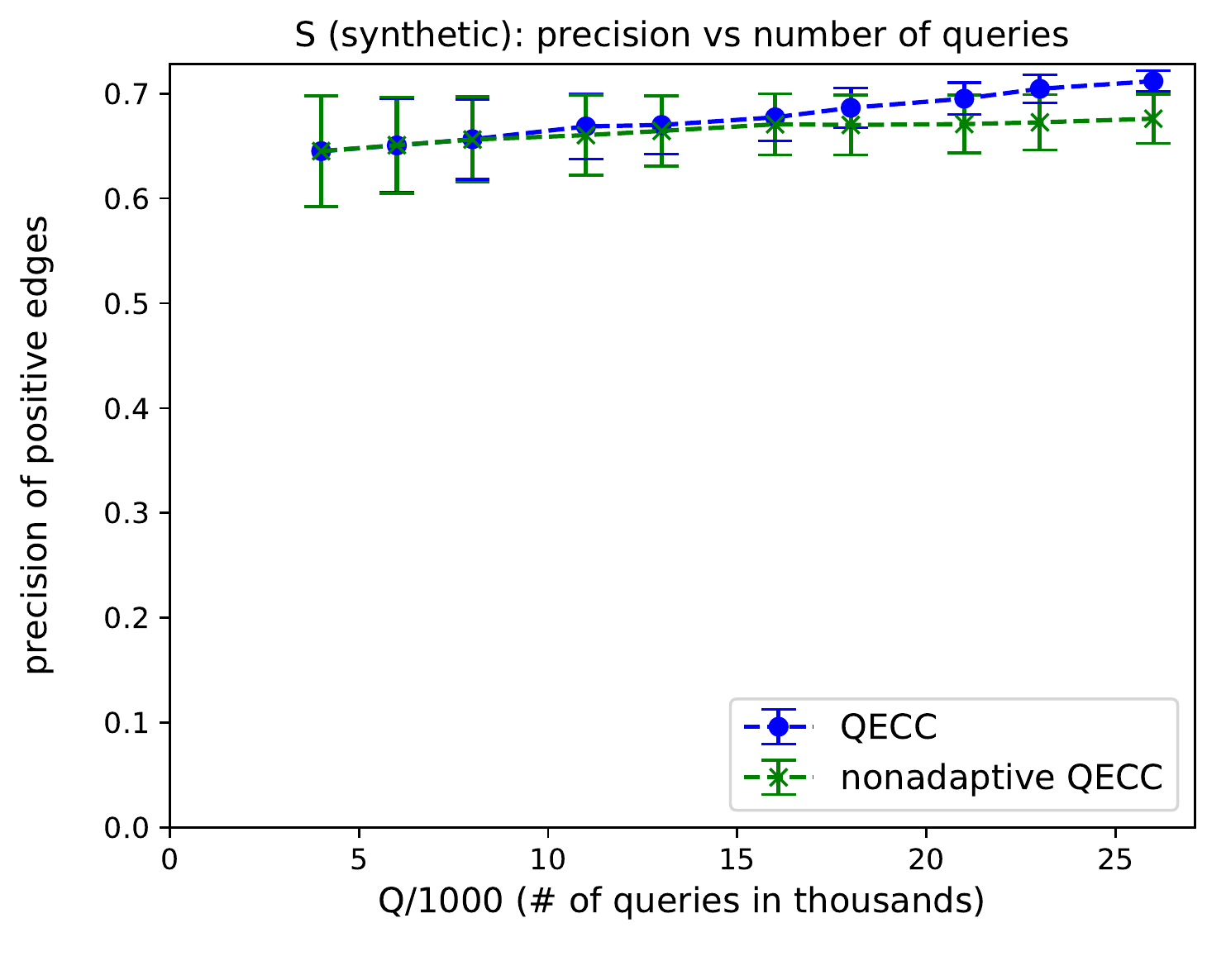}\\\
\includegraphics[width=0.32\textwidth]{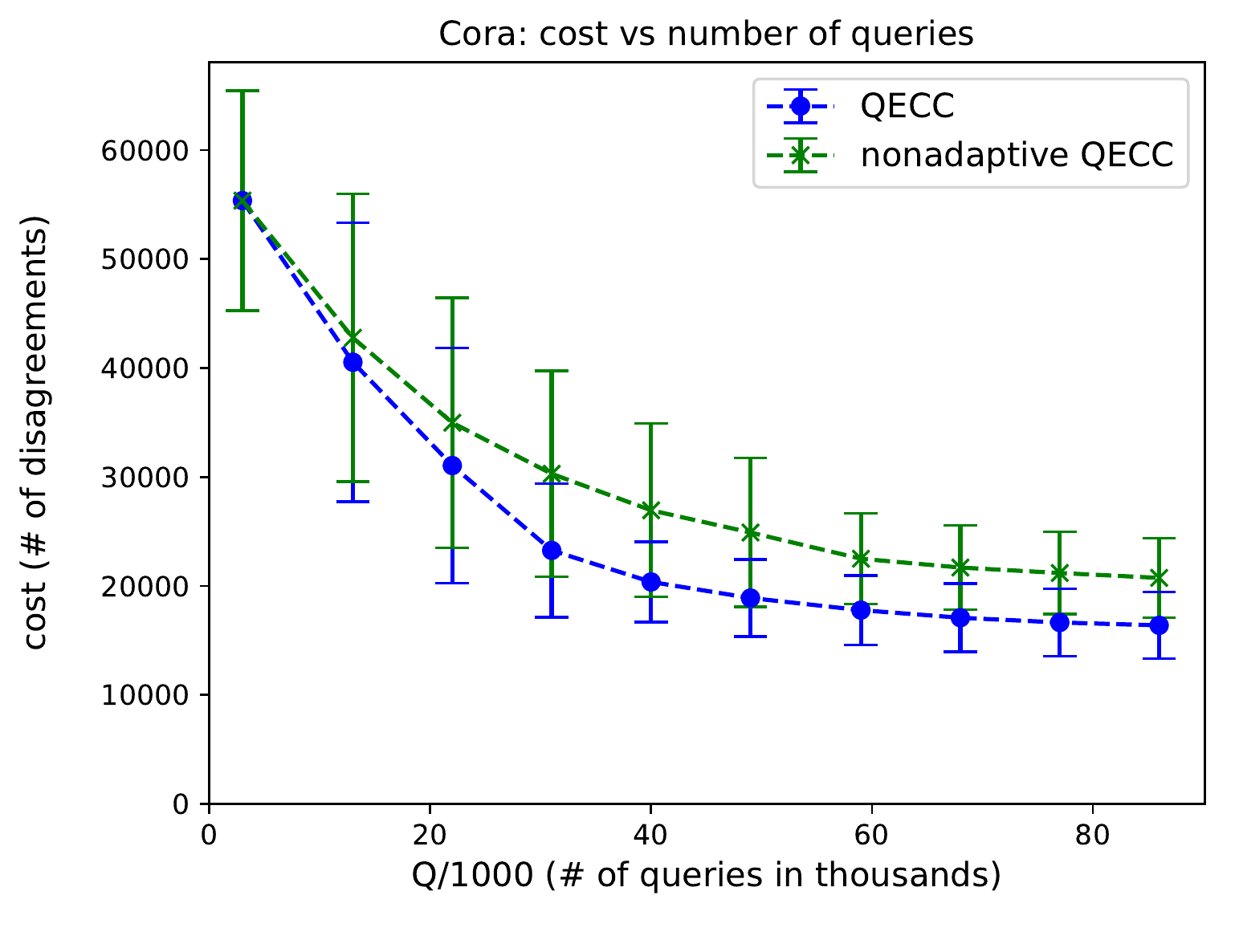} & \includegraphics[width=0.31\textwidth]{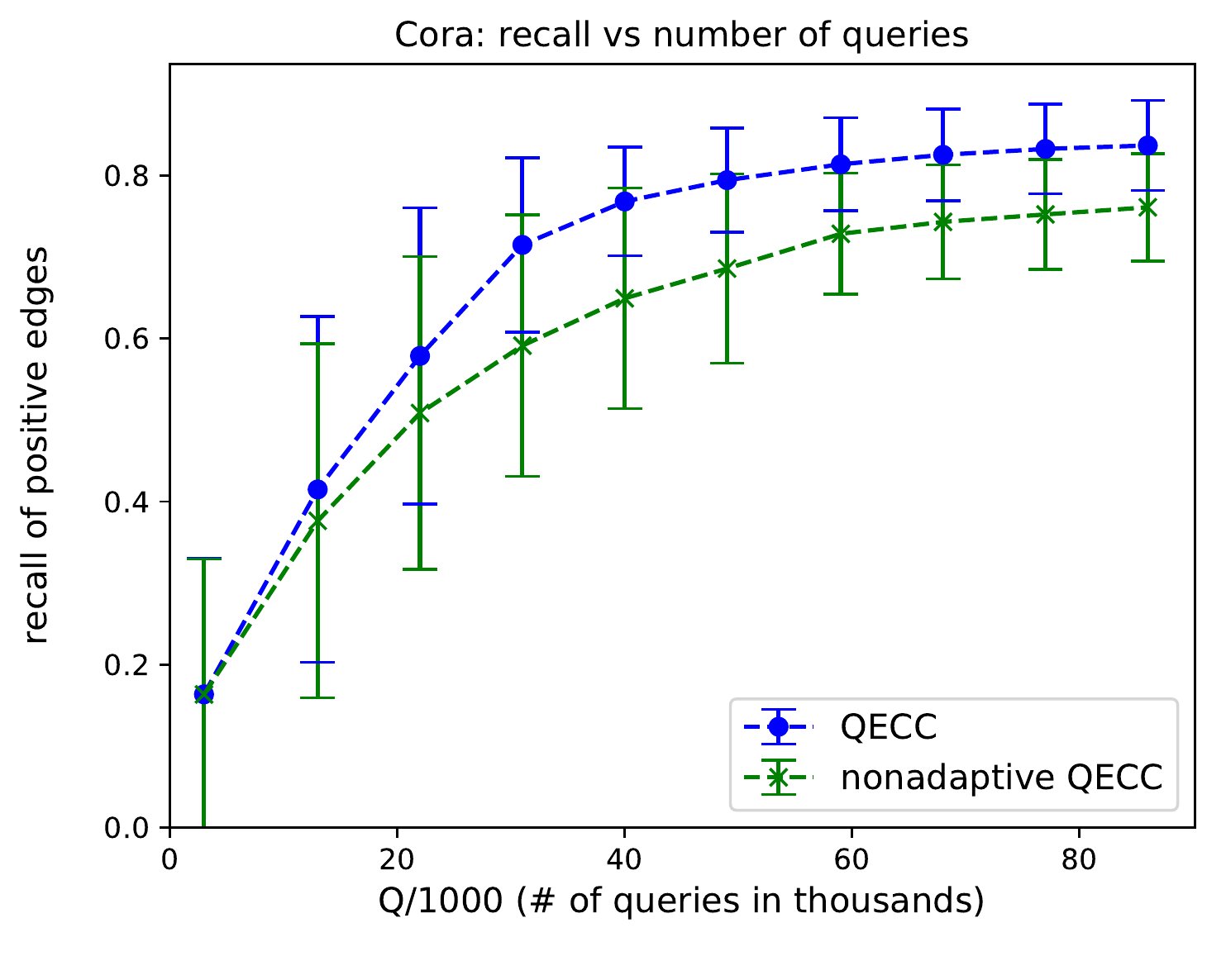} & \includegraphics[width=0.31\textwidth]{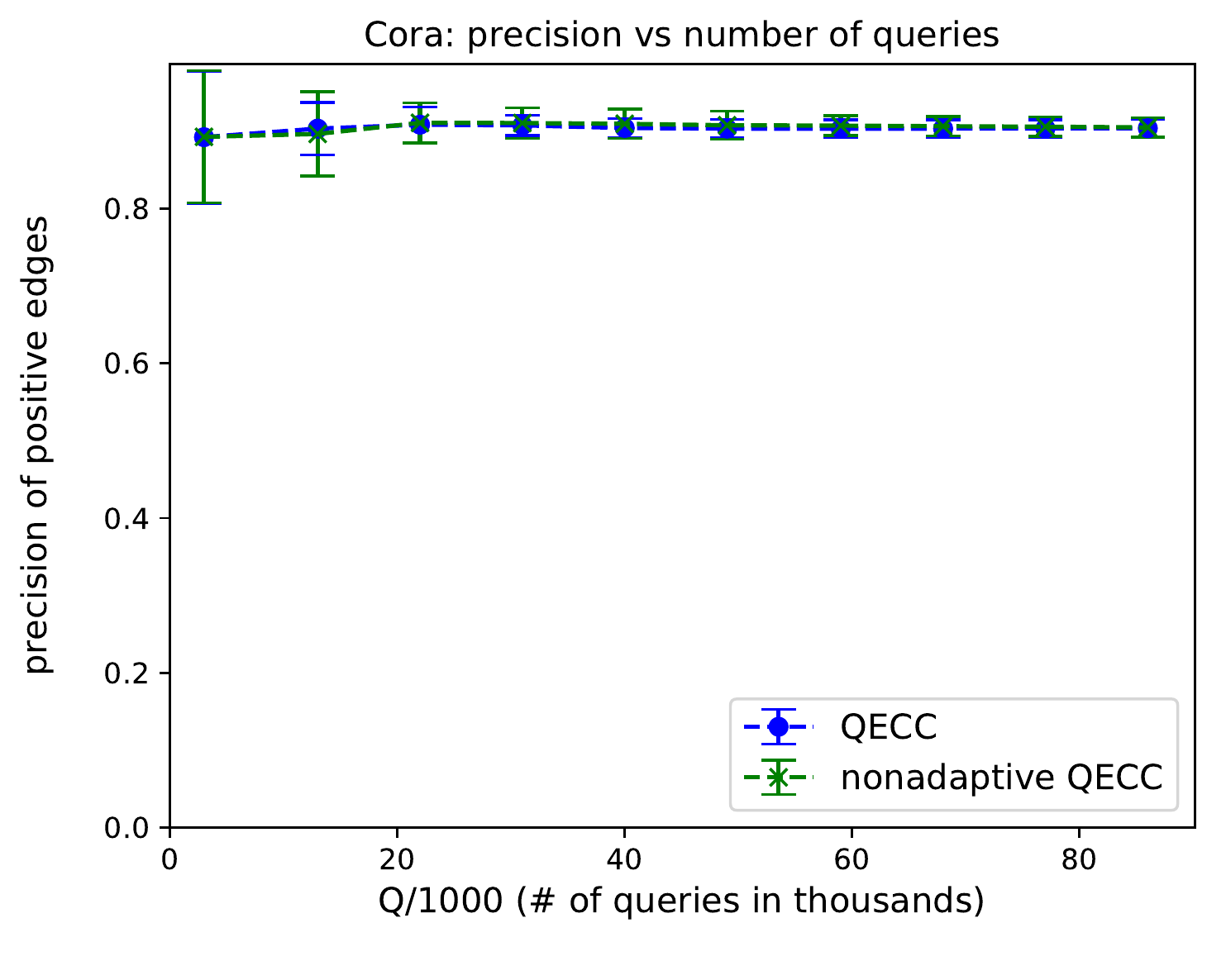}\\
  \end{tabular}
\caption{\label{fig:nonadaptive}
        Comparison of $\ouralg$ and its non-adaptive variant.
}
\end{figure*}
Table~\ref{tab:datasets} summarizes the datasets we tested.
Figure~\ref{fig:findings} shows the measured clustering cost against the number of queries $Q$ performed by $\ouralg$ and $\ourheur$ in the synthetic graph $S(2000, 20, 0.15, 2)$
and the real-world $\textsc{Cora}$, $\textsc{Citeseer}$ and
$\textsc{Mushrooms}$ datasets.

\spara{Comparison with the baseline.} It is clearly visible that both $\ourheur$ and $\ouralg$ perform noticeably better than the baseline for all query budgets $Q$.
As expected, all accuracy measures are improved with higher query budgets.
The number of non-singleton clusters found by $\ourheur$ and $\ouralg$ increases with higher values of $Q$, but decreases when using
the affinity-propagation-based baseline.
We do not show this value for the baseline on Mushrooms because it is of the order of hundreds; in this case the ground truth number of clusters is just two,
 and  $\balls$, $\ouralg$ and $\ourheur$ need very few queries (compared to $n$) to find the clusters quickly.

\spara{$\balls$ vs $\ouralg$ and $\ourheur$.}
At the limit, where $Q$ equals the average
number $A$ of queries made by $\balls$, both $\ouralg$ and $\ourheur$ perform as well as $\balls$.
%and, interestingly, find a number of clusters close to that of the ground truth.
In our synthetic dataset, the empirical average cost of $\balls$
is roughly 2.3 times the cost of the ground truth, suggesting that it is nearly a worst-case instance for our algorithm since $\balls$ has an expected 3-approximation guarantee.
Remarkably, in the real-world dataset
$\textsc{cora}$, $\ourheur$ can find a solution just as good as the ground truth with just 40,000 queries. Notice that this is half what $\balls$ needs and much smaller than the
$\binom{n}{2} \approx 1.7$ million queries that full-information methods  for correlation clustering such as~\cite{near_opt} require.

\spara{Effect of graph characteristics.}
As expected, total cost and recall improve with the number of queries on all datasets (Figure~\ref{fig:findings}); precision, however, remains mainly constant throughout
a wide range of query budgets.
To evaluate the impact of the graph and noise parameters on the performance of our algorithms, we perform additional tests on synthetic datasets where we fix all parameters to those of $S(2000, 20, 0.15,
        2)$ except for the one under study. Figure~\ref{fig:variation}
    shows the effect of graph size $n$ (1st row), number of clusters $k$, (2nd row), imbalance parameter $\alpha$ (3rd row) and noise parameter $\beta$ (4th row) on
    total cost and recall, in synthetic datasets. Here we used $Q = 15000$ for all three query-bounded methods: $\ouralg$, $\ourheur$ and the baseline.
    Naturally, $\balls$ gives the best results as it has no query limit. All other methods tested follow the same trends, most of which are intuitive:
\begin{itemize}
  \item  Cost increases with $n$, and recall decreases, indicating that more queries are necessary in larger graphs to achieve the same quality.
    Precision, however stays constant.
  \item Cost decreases with $k$ (because the graph has fewer positive edges). Recall stays constant for the ground truth and the unbounded-query method $\balls$ as it is
    essentially determined by the noise level, but it decreases
    with $k$ for the query-bounded methods. Again, precision remains constant except for the baseline, where it decreases with $k$.
  \item Recall increases with imbalance $\alpha$ because the largest cluster, which is the easiest to find, accounts for a larger fraction of the total number of edges $m$.
    Precision also increases.
    On the other hand, $m$ itself increases with imbalance, possibly explaining the increase in total cost.

\item Finally, cost increases linearly with the level of noise $\beta$, while recall and precision decrease as $\beta$ grows higher.
\end{itemize}

\spara{Effect of adaptivity.} Finally, we compare the adaptive $\ouralg$ with Non-adaptive $\ouralg$ as described at the end of Section~\ref{sec:main}.
Figure~\ref{fig:nonadaptive} compares the performance of both on the synthetic dataset and on $\textsf{Cora}$.
While both have the same theoretical
guarantees, it can be observed that the non-adaptive variant of $\ouralg$ comes at moderate increase in cost and, decrease in recall and precision.
%The adaptive algorithm can find more pivots (cluster centers) because it does not need to waste queries between a newly found pivot and the neighbors of previously found pivots.
\mycomment{

TODO
    Plots, real data:
        Balls, QEAffinity, QECC, heuristic with Q ranging:
            cora: 1 with error, 1 with recall, 1 with precision
            newsgroups? 1 with error, 1 with precision
    Plots,

    - plots, fix parameters, how?
    - some larger graph? string/Youtube?
    - plots with color names
    - plots with average number of clusters?
    \mycomment{
   datasets:
   pubmed https://linqs.soe.ucsc.edu/data
   NELL dataset can be found here at http://www.cs.cmu.edu/~zhiliny/data/nell_data.tar.gz
   https://linqs.soe.ucsc.edu/data

   from https://github.com/tkipf/gcn:

In this example, we load citation network data (Cora, Citeseer or Pubmed). The original datasets can be found here: http://linqs.cs.umd.edu/projects/projects/lbc/. In our version (see data folder) we use dataset splits provided by https://github.com/kimiyoung/planetoid (Zhilin Yang, William W. Cohen, Ruslan Salakhutdinov, Revisiting Semi-Supervised Learning with Graph Embeddings, ICML 2016).
    }
}

\section{Conclusions}\label{sec:conclusions}
This paper presents the first query-efficient correlation clustering algorithm with provable guarantees.
The trade-off between the running time of our algorithms
and the quality of the solution found  is nearly optimal. We also presented a more practical algorithm that consistently achieves higher recall values than  our theoretical
algorithm.
Both of our algorithms are amenable to simple implementations.

A natural question for further research would be to
obtain query-efficient algorithms based on the better LP-based approximation algorithms~\cite{near_opt}, improving the constant factors in our guarantee.  Another intriguing question is whether one can devise other graph-querying models that allow for improved theoretical results  while being reasonable from a practical viewpoint. The
reason
an additive term is needed in the error bounds is that, when the graph is very sparse, many queries are needed to distinguish it from an empty graph
(i.e., finding a positive edge).
We note
that if we allow neighborhood oracles (i.e., given $v$, we can obtain a linked list of the \emph{positive} neighbours of~$v$ in time linear in its length), then we can derive a
constant-factor approximation algorithm with $O(n^{3/2})$ neighborhood  queries, which can be significantly smaller than the number of edges.
Indeed, Ailon and Liberty~\cite{correlation_revisited} argue that with a neighborhood oracle, $\balls$ runs in time $O(n +
        OPT)$; if $OPT \le n^{3/2}$ this is $O(n^{3/2})$. 
On the other hand, if $OPT > n^{3/2}$ we can stop the algorithm after $r = \sqrt n$ rounds, and by Lemma~\ref{lem:indep}, we incur an additional cost of only $O(n^{3/2}) =
O(OPT)$.
This shows that more powerful oracles  allow for smaller query complexities.
Our heuristic $\ourheur$ also suggests that granting the ability to query a random positive edge may help.
These questions are particularly relevant to clustering graphs with many small clusters.

\begin{acks}
Part of this work was done while CT was visiting ISI Foundation. DGS, FB, and CT acknowledge support from Intesa Sanpaolo Innovation Center. The funders had no role in study design, data collection and analysis, decision to publish, or preparation of the manuscript.
\end{acks}

\bibliographystyle{ACM-Reference-Format}
\bibliography{bibliography}

%\medskip 

%\begin{mdframed}[innerbottommargin=3pt,innertopmargin=3pt,innerleftmargin=6pt,innerrightmargin=6pt,backgroundcolor=gray!10,roundcorner=10pt]
%\spara{Acknowledgments.} Part of this work was done while CT was visiting ISI Foundation. DGS, FB, and CT acknowledge support from Intesa Sanpaolo Innovation Center. The funders had no role in study design, data collection and analysis, decision to publish, or preparation of the manuscript.
%\end{mdframed} 

\end{document}